\documentclass[10pt, a4paper]{article}

\usepackage[dvips]{graphicx}
\usepackage[top=2.4cm, bottom=2.4cm, left=2cm, right=2cm]{geometry}
\usepackage{amssymb}
\usepackage{xspace}
\usepackage{graphicx}
\usepackage{subfigure}
\usepackage{xspace}
\usepackage{amsmath}
\usepackage{amsfonts}
\usepackage{algorithm}
\usepackage{algpseudocode}
\usepackage{multirow}
\usepackage{setspace}
\usepackage{epstopdf}
\usepackage{multirow}

\newcommand{\secref}[1]         {Section~\ref{sec:#1}}

\newcommand{\figref}[1]         {Figure~\ref{fig:#1}}
\newcommand{\figreftwo}[2]         {Figures~\ref{fig:#1} and~\ref{fig:#2}}

\newcommand{\tabref}[1]         {Table~\ref{tab:#1}}
\renewcommand{\eqref}[1]          {Equation~(\ref{eq:#1})}
\newcommand{\eqreftwo}[2]        {Equations~(\ref{eq:#1}) and~(\ref{eq:#2})}

\newcommand{\calset}[1]{\ensuremath{\mathop{\mathcal{#1}}\nolimits}}
\newcommand{\jobset}{\calset{J}}
\newcommand{\machineset}{\calset{M}}

\newcommand{\slotset}{\calset{S}}
\newcommand{\A}{\calset{A}}
\newcommand{\Q}{\calset{Q}}

\newcommand{\paren}[1]          {\left( #1 \right)}

\renewcommand{\algorithmicrequire}{\textbf{Input:}}
\renewcommand{\algorithmicensure}{\textbf{Output:}}
\DeclareMathOperator*{\argmin}{arg\!\min}
\renewcommand*\bar[1]{\rlap{$\smash{\overline{#1}}$}\phantom{#1}}
\newtheorem{claim}            {Claim}
\newenvironment{proof}      {\noindent{\em Proof.}\hspace{1em}}{\qed}
\def\squarebox#1{\hbox to #1{\hfill\vbox to #1{\vfill}}}
\newcommand{\qedbox}            {\vbox{\hrule\hbox{\vrule\squarebox{.667em}\vrule}\hrule}}
\newcommand{\qed}               {\nopagebreak\mbox{}\hfill\qedbox\smallskip}

\def\abstract{\vspace{.5em}
{\textit{\bf Abstract}:\,\relax%
}}

\def\keywords{\vspace{.5em}
{\textit{\bf Keywords}:\,\relax%
}}
\def\endkeywords{\par}

\begin{document}

\title{Energy-Efficient and Thermal-Aware Resource Management for Heterogeneous Datacenters}
\author{Hongyang Sun, Patricia Stolf, Jean-Marc Pierson, and Georges Da Costa \\
\\
IRIT, University of Toulouse\\
118 Route de Narbonne, F-31062 Toulouse Cedex 9, France\\
\{sun, stolf, pierson, dacosta\}@irit.fr
}
\date{}
\maketitle

\begin{abstract}
We propose in this paper to study the energy-, thermal- and performance-aware resource management
in heterogeneous datacenters. Witnessing the continuous development of heterogeneity in
datacenters, we are confronted with their different behaviors in terms of performance, power
consumption and thermal dissipation: Indeed, heterogeneity at server level lies both in the
computing infrastructure (computing power, electrical power consumption) and in the heat removal
systems (different enclosure, fans, thermal sinks). Also the physical locations of the servers
become important with heterogeneity since some servers can (over)heat others. While many studies
address independently these parameters (most of the time performance and power or energy), we show
in this paper the necessity to tackle all these aspects for an optimal resource management of the
computing resources. This leads to improved energy usage in a heterogeneous datacenter including
the cooling of the computer rooms. We build our approach on the concept of heat distribution matrix
to handle the mutual influence of the servers, in heterogeneous environments, which is novel in
this context. We propose a heuristic to solve the server placement problem and we design a generic
greedy framework for the online scheduling problem. We derive several single-objective heuristics
(for performance, energy, cooling) and a novel fuzzy-based priority mechanism to handle their
tradeoffs. Finally, we show results using extensive simulations fed with actual measurements on
heterogeneous servers.

\end{abstract}

\section{Introduction}

The last years have witnessed the development of heterogeneity in clusters and datacenters. Two
main reasons have led to this situation today. The first one is due to the maintenance and
evolution of the components in the datacenters: different generations of computers are commonly
seen in production datacenters since the owners are not changing everything at each update. The
second reason is driven by the idea that heterogeneity might be the key to achieving
energy-proportional computing \cite{Barroso07_Proportional, DaCosta13_Heterogeneity}, especially
for high-performance computing applications.

Many recent studies alert dramatically on the energy consumption of the datacenters. For instance,
Koomey's report \cite{Koomey08_DCReport} claims that today's datacenters are consuming nearly 2\%
of the global energy, and up to half of that is spent on cooling-related activities
\cite{Sawyer04_coolingCost}. This results generally in very poor Power Usage Effectiveness (PUE).

In this paper, we study the multi-objective resource management problem for heterogeneous
datacenters. Besides the performance criterion, we also consider the energy consumption of the
servers and their thermal impact on the datacenter cooling. The aim of our work is to optimize
these objectives and to explore their tradeoffs. In particular, the energy consumption is partly
due to the cooling efficiency in the datacenter \cite{Moore05_Datacenter, Tang08_CyberPhysical},
which is related to both the physical placement of the servers and the scheduling strategies when
jobs dynamically enter and leave the system. The latter also affects the performance and the energy
consumed by the servers.

Server placement in a computer room has been relatively less studied, especially its impact on the
cooling efficiency. The reason for this lack of attention is mainly due to the fact that, when
servers are homogeneous, their relative positions have no impact on the performance and computing
energy. However, server placement can have an impact on the cooling infrastructure. The main
observation is that one server might contribute to the temperature raise at the inlets of the other
servers, due to the recirculation of heat in a datacenter. Such mutual influence can be modeled by
a heat distribution matrix among the servers. If one wants to keep the inlet temperature under a
given threshold, the supplied air temperature has to be decreased accordingly by the cooling
system, which in turn increases its energy consumption. In the presence of heterogeneous servers
with different power consumptions and hence heat dissipation, the problem of find the optimal
placement becomes complicated and, to the best of our knowledge, has not been studied. Since it is
not feasible to change dynamically the positions of the servers in a datacenter, we focus on static
placement to minimize the cooling cost induced by different configurations.

With a given server placement, the traditional problem of job scheduling in the heterogeneous
environment remains. Many previous work (e.g., \cite{Barbosa09_Heterogeneous,
Topcuoglu02_Heterogeneous}) considered only the performance criterion and hence focused on the
jobs' execution times. In order to address the power consumption issue in datacenters, however,
application scheduling must employ a multi-objective approach by considering performance, energy
and cooling together. To account for the fact that a scheduler has no future knowledge (jobs arrive
over time), we need an online scheduling strategy. Instead of designing different independent
algorithms, we design a greedy online scheduling framework that can be adapted easily by redefining
the cost function, from a single objective to two or more objectives. To tackle the
energy-performance tradeoff, we further introduce a fuzzy-based priority approach, which allows to
explore the potential improvement in one objective while relaxing the other objective up to an
acceptable range. This approach can be extended to incorporate more than two objectives in the
framework. Its principle is not limited to the case at hand and can potentially be applied to other
multi-objective optimization problems.

The main contributions of this paper are the following:
\begin{itemize}
\item A static server placement heuristic to reduce the cooling cost for the servers in a datacenter.
\item A greedy scheduling framework and several cost functions to tackle single-objective scheduling (for performance, energy, and
cooling).
\item A fuzzy-based priority approach to handle the tradeoff between two conflicting objectives, and its extension to
multi-objective optimization.
\end{itemize}

These proposals are supported by extensive simulations conducted using real hardware specifications
and software benchmarks, as well as experimentally verified cooling model and heat distribution
matrix \cite{Tang06_SensorBased, Tang08_CyberPhysical}. Specifically for the hardware, a server
system with high packing density and integrated cooling support is chosen for the experiments,
which we believe represents well an emerging class of highly integrated energy-efficient servers.
The results demonstrate the flexibility of our scheduling framework and confirm the effectiveness
of the fuzzy-based approach for exploring the energy-performance tradeoff in heterogeneous
datacenter environments. Our static server placement heuristic is also shown to provide much
improved thermal distribution leading to significant reduction in cooling cost.

The rest of this paper is organized as follows. \secref{statement} formally states the system model
and the scheduling problem. \secref{gmp} describes our greedy server placement heuristic.
\secref{scheduling} presents the job scheduling framework, various cost functions and the
fuzzy-based priority approach. The simulation results are shown in \secref{evaluation}.
\secref{related} reviews some related work, and \secref{conclusion} summarizes the paper and
addresses future directions.

\section{Problem Statement}\label{sec:statement}

\subsection{System Model}

Motivated by the placement of physical servers and the scheduling of High-Performance Computing
(HPC) applications in heterogeneous datacenters, we consider the following system model: A set
$\machineset = \{M_1, M_2, \cdots, M_m\}$ of $m$ servers (or machines) needs to be placed inside a
computer room (or datacenter) with a set of $m$ rack slots, denoted by  $\slotset = \{S_1, S_2,
\cdots, S_m\}$\footnote{In this paper, we assume that the number of rack slots is equal to the
number of servers to be placed, which represents a common scenario in small- and medium-size
datacenters.}. Each server $M_j \in \machineset$ consists of $L_j$ processors of the same type
(possibly on different boards), but the type and the number of processors may vary for different
servers, rendering the system heterogeneous. Each server consumes a \emph{base power} $U^{base}_j$
to support the basic operations of the infrastructure backbone, such as monitoring, networking and
cooling (for instance fans). A set $\jobset = \{J_1, J_2, \cdots, J_n\}$ of $n$ jobs arrive at the
system over time, and they need to be assigned in an online manner to the servers. Each job $J_i
\in \jobset$ has a release time $r_i$ and a processor requirement $l_i$ that must be granted in
order to run on any server. To execute job $J_i$ on server $M_j$ incurs a processing time $P_{i,
j}$ and a power consumption $U_{i, j}$, both of which are server-dependent and become known upon
the job's arrival by prior profiling of the applications. In particular, the profiled application
power consumption is assumed to include the leakage power.

\subsection{Scheduling Model}

We study two orthogonal problems that deal with the placements of hardware and software,
respectively. We call the two problems \emph{static server placement} and \emph{online job
scheduling}. The former concerns the positioning of physical servers in the datacenter, which as
explained in \secref{gmp} will have an impact on the cooling energy in heterogeneous environment.
The latter concerns the dynamic assignment of workloads to the servers, which will impact energy
(due to both computing and cooling) as well as performance.

For the first problem of static server placement, each server needs to be physically and statically
placed in advance to one of the available rack slots in the datacenter. In particular, we are
looking for a mapping $\sigma:\{1, 2, ,\cdots, m\}\rightarrow \{1, 2, ,\cdots, m\}$ from rack slots
to servers so that each slot $S_k$ is filled with a server $M_{\sigma(k)}$. The objective is to
minimize the cooling cost. More details about this problem will be described in \secref{gmp}.

Given a particular server placement, an online scheduling strategy is then required to assign the
jobs to the servers for execution. Specifically, each arrived job $J_i \in \jobset$ must be
assigned irrevocably to a server with at least $l_i$ idle processors, and without any knowledge of
the future arriving jobs. Once the job has been assigned, no preemption or migration is allowed,
which is typically assumed for HPC applications since they tend to incur a significant cost in
terms of data reallocation.

At any time $t$, the total computing power of server $M_j$ is the sum of its base power and the
power consumed for executing all jobs assigned to it, i.e.,
\begin{equation}
U^{comp}_j(t) = U^{base}_j + \sum_{i = 1}^{n} \delta_{i, j}(t)\cdot U_{i, j} \ ,
\end{equation}
where $\delta_{i, j}(t)$ is a binary variable that takes value $1$ if job $J_i$ is running on
server $M_j$ at time $t$ and $0$ otherwise. In order to optimize performance, we do not allow
processor sharing among the jobs. Thus, each server at any time can only host a subset of the jobs
whose total processor requirements are no more than the server's total number of available
processors, i.e., $\sum_{i=1}^n\delta_{i, j}(t)\cdot l_i \le L_j$ for all $1\le j\le m$ at all time
$t$.

\subsection{Cooling Model}

To characterize the cost of cooling, we consider a standard datacenter layout, where server racks
are organized in rows with alternating cold and hot aisles. The computer room air conditioning
(CRAC) unit supplies cool air to the cold aisles through raised floor vents. Each server $M_j \in
\machineset$ in the racks is oriented such that it draws cool air with temperature $T_j^{in}$ from
the inlet and dissipates hot air with temperature $T_j^{out}$ to the outlet. Assuming that the
computing power consumed by a server is completely transformed into heat, the relationship between
the power consumption and the inlet/outlet temperature of server $M_j$ at any time $t$ can be
characterized by \cite{Tang06_SensorBased}:
\begin{equation}\label{eq:energyConserve}
T^{out}_j(t) = T^{in}_j(t) + K_j \cdot U^{comp}_j(t) \ ,
\end{equation}
where $K_j = p f_j c$, with $p$ denoting the air density (in $kg/m^3$), $f_j$ the airflow rate of
server $M_j$ (in $m^3/s$), and $c$ the air heat capacity\footnote{The \emph{air heat capacity}
specifies the energy required to change the temperature of one unit mass of air by one unit
degree.} (in $J/(^{\circ}{\rm C}\cdot kg)$).

Due to complex airflow patterns, typical datacenters experience the so-called \emph{heat
recirculation} phenomenon, where the hot air from the server outlets recirculates in the room and
is mixed with the supplied cool air from the CRAC, causing the temperature at the server inlets to
be higher than that of the supplied air. Prior studies \cite{Tang06_SensorBased,
Tang08_CyberPhysical} have characterized this phenomenon with a \emph{heat distribution matrix}
$\mathbf{D}$ by assuming a fixed airflow pattern in the room and conservation of energy as
described by \eqref{energyConserve}. We adopt this approach here. Let each element $d_{j, k} \in
\mathbf{D}$ represent the temperature increase at the inlet of server $M_j$ per unit of power
consumed by server $M_k$.\footnote{Technically speaking, $d_{j, k}$ represents the temperature
increase for the server at slot $S_j$ due to the power consumption by the server at slot $S_k$. For
convenience, we simply assume that the servers are renamed such that server $M_j$ is placed in slot
$S_j$ for all $1\le j \le m$.} Combining the heat contributions from all servers, the inlet
temperature of server $M_j$ at time $t$ is given by the following equation:
\begin{equation}\label{eq:tin}
T^{in}_j(t) = T^{sup}(t) + \sum_{k=1}^{m}d_{j, k}\cdot U^{comp}_k(t) \ ,
\end{equation}
where $T^{sup}(t)$ denotes the supplied air temperature at time $t$, which should be adjusted to
prevent the inlet temperature of any server from going beyond a \emph{redline temperature}
$T^{red}$; otherwise, the electronic components may not work reliably or are at risk of being
damaged. Hence, the supplied air temperature should be set at most to
\begin{equation}\label{eq:supplyTemp}
T^{sup}(t) = T^{red} - \max_{j=1..m}\sum_{k=1}^{m}d_{j, k}\cdot U^{comp}_k(t) \ .
\end{equation}
The cooling cost is specified as
\begin{equation}\label{eq:coolingCost}
U^{cool}(t) = \frac{\sum_{j=1}^{m}U_j^{comp}(t)}{\mbox{CoP}(T^{sup}(t))} \ ,
\end{equation}
where CoP is the \emph{coefficient of performance}, defined as the ratio of the amount of heat to
be removed to the energy that needs to be consumed in order to perform the cooling
\cite{Moore05_Datacenter}. This coefficient characterizes the efficiency of the CRAC unit, and is
an increasing (usually non-linear) function of the supplied air temperature. Intuitively, it means
that the CRAC unit needs to work harder and thus consumes more energy in order to provide cooler
air to the computer room.

\subsection{Optimization Objectives}

We consider the following \emph{bi-objective optimization problem}: optimizing the performance of
the jobs and minimizing the energy consumption of the datacenter, due to both computing and
cooling.\footnote{The energy consumed by other parts of the datacenter, such as lighting, are
ignored, since they are insignificant compared to the computing and cooling energy.}

For performance, we use the average response time of the jobs as the metric, and it is defined as
\begin{equation}
R_{ave} = \frac{1}{n}\sum_{i=1}^{n}(c_i - r_i) \ ,
\end{equation}
where $c_i$ and $r_i$ denote the completion time and release time of job $J_i$, respectively.

The energy consumption comes from two sources: computing and cooling. The one due to computing is
given by the total computing power of all servers integrated over time, i.e.,
\begin{equation}
E_{comp} = \int_{t_1}^{t_2} \sum_{j=1}^{m}U_j^{comp}(t)dt \ ,
\end{equation}
where $[t_1, t_2]$ denotes the interval of interest, during which all jobs arrive and complete
their executions. This computing energy can be further divided into two parts, namely, the
\emph{static} part due to the base power consumption, i.e.,
\begin{equation}
E_{comp}^{stat} = (t_2-t_1)\cdot \sum_{j=1}^{m} U_j^{base} \ ,
\end{equation}
and the \emph{dynamic} part due to the power consumed for executing the jobs, i.e.,
\begin{equation}
E_{comp}^{dync} = \sum_{i=1}^{n}\sum_{j=1}^{m} \delta_{i,j} \cdot P_{i, j}\cdot U_{i, j} \ ,
\end{equation}
where $\delta_{i, j} = 1$ if job $J_i$ is assigned to server $M_j$ and 0 otherwise.

The energy spent on cooling is the total cooling power integrated over time, i.e.,
\begin{equation}
E_{cool} = \int_{t_1}^{t_2} U^{cool}(t)dt \ ,
\end{equation}
and as with computing energy, cooling energy can also be broken into a static part and a dynamic
part. Specifically, the static part is the cooling energy that will be spent during interval $[t_1,
t_2]$ even if no job arrives, i.e.,
\begin{equation}
E_{cool}^{stat} = \int_{t_1}^{t_2}\frac{\sum_{j=1}^{m}U_j^{base}(t)}{\mbox{CoP}(T^{red} -
\max_{j}\sum_{k}d_{j, k}\cdot U^{base}_k(t))} dt \ ,
\end{equation}
and the dynamic part is the difference between the total cooling energy and the static one, i.e.,
\begin{equation}
E_{cool}^{dync} = E_{cool} - E_{cool}^{stat} \ .
\end{equation}

In this paper, we assume that all servers are turned on all the time to sustain the servers'
infrastructure backbone, so the static energy due to both computing and cooling is independent of
the workload and the job scheduling strategy. On the other hand, the total dynamic energy given by
\begin{equation}\label{eq:totalDynamic}
E_{total}^{dync} = E_{comp}^{dync} + E_{cool}^{dync}
\end{equation}
is closely related to job scheduling, and it will be the focus of this study.

Due to the heterogeneity of the servers in the datacenter, different job scheduling strategies may
result in very different job response time, computing energy and cooling cost. While a specific
scheduling strategy may optimize one objective, these different objectives can be conflicting with
each other, making the optimization difficult. In \secref{scheduling}, we will propose and evaluate
online scheduling algorithms to address both performance and energy as well as to deal with their
tradeoffs.

\section{Static Server Placement and A Greedy Heuristic}\label{sec:gmp}

In this section, we consider the problem of static server placement. We first motivate the study
from the perspective of cooling in heterogeneous datacenters. We then formulate the problem and
present a greedy heuristic.

\subsection{Motivation}\label{sec:motivation}

The literature contains extensive studies on virtual machine placement (e.g.,
\cite{Borgetto12_VirtualMachine, Gao13_VirtualMachines, Xu10_MultiFuzzySet}) for datacenters, but
the placement of physical servers has received little attention. There are two main reasons. First,
many traditional datacenters are homogeneous, so different placements of identical servers do not
make a difference. Second, traditional metrics such as job performance and energy consumption (due
to computing) are independent of the servers' relative positions, so they are unaffected by the
different placement configurations.

As far as the cooling cost is concerned for heterogeneous datacenters, however, the placement of
the physical servers will have an impact. In particular, the studies in \cite{Tang06_SensorBased,
Tang08_CyberPhysical} have shown that the heat recirculation phenomenon in typical datacenters
exhibits the following properties:

(1). Different rack positions tend to behave differently in terms of heat recirculation. Typically,
servers located at the upper parts of the racks ``inhale" more recirculated hot air while servers
located at the lower parts ``contribute" more hot air to recirculate in the room.

(2). In a closed computer room with fixed locations of all major objects and without moving
objects, the airflow pattern that characterizes the heat recirculation among different rack
positions is relatively stable.

While the first property suggests that the heat distribution matrix tends to be highly asymmetric,
the second property assures that the matrix does not change significantly with different workloads
in the servers or different positions of the servers. In the next section, we will rely on workload
placement (or job scheduling) techniques to manage the cooling cost together with other objectives.
Here, we focus on arranging the positions of the servers with different power profiles. The goal is
to reduce the maximum inlet temperature of the servers so as to minimize the cooling cost under a
given load condition.

To illustrate the effectiveness of this approach, consider a simple datacenter with two servers,
two rack slots, and the following heat distribution matrix:
\begin{equation*}
\mathbf{D} = \begin{bmatrix}
  0.002 & 0.004 \\
  0.001 & 0.002
\end{bmatrix} .
\end{equation*}
Suppose the two servers consume an average power of $100$W and $200$W, respectively. By placing the
first server in slot 1 and the second server in slot 2, their inlet temperatures increase by
$1^{\circ}{\rm C}$ and $0.5^{\circ}{\rm C}$ respectively according to \eqref{tin}. By simply
swapping the positions of the two servers, their temperature increases will now become
$0.4^{\circ}{\rm C}$ and $0.8^{\circ}{\rm C}$. The $0.2^{\circ}{\rm C}$ difference in the maximum
inlet temperature of these two configurations directly determines the temperature of the supplied
air by \eqref{supplyTemp}, and therefore impacts the cooling cost. For instance, consider a redline
temperature of $25^{\circ}{\rm C}$ and the following CoP model for a water-chilled CRAC unit in an
HP datacenter \cite{Moore05_Datacenter, Tang08_CyberPhysical}:
\begin{equation}\label{eq:hp_cop}
\mbox{CoP}(T) = 0.0068T^2 + 0.0008T + 0.458 \ .
\end{equation}
According to \eqreftwo{supplyTemp}{coolingCost}, the cooling costs for the two placement
configurations are $68.275$W and $67.269$W, respectively. The impact will be more significant with
a lower redline temperature or a more skewed heat distribution matrix, or when the servers are
consuming more power. The problem will also become more challenging when there is a large number of
servers/positions, since exhaustive search will no longer be possible. The next subsection
considers this general case and proposes a heuristic algorithm for the problem.

\subsection{Greedy Heuristic}\label{sec:heuristic}

To reduce the cooling cost, we should minimize the maximum temperature increase at the inlet of any
server in the datacenter. As we have seen previously, this is determined by both the
heat-distribution matrix and the power consumption profile of all servers. While the former is
relatively stable and can be measured using a sensor-based approach \cite{Tang06_SensorBased}, the
latter essentially depends on the servers' workloads, which can vary with time. To cope with this
uncertainty, we characterize the power consumption of each server statically using the average
power it consumes when executing historical workloads. This provides a reasonable estimation on the
server's typical power consumption during runtime. We call this static value the \emph{reference
power}, and use it to determine the placement of the servers.

Let $U^{ref}_{j}$ denote the reference power of server $M_j \in \machineset$. The \emph{static
server placement} problem can then be formulated as follows: find a mapping $\sigma:\{1, 2,
,\cdots, m\}\rightarrow \{1, 2, ,\cdots, m\}$ from rack slots to servers, so as to
\begin{equation}
\mbox{minimize  }\max\mbox{ }\mathbf{D}\cdot\mathbf{U}_{\sigma}^{ref} \ ,
\end{equation}
where $\mathbf{U}_{\sigma}^{ref} = [U^{ref}_{\sigma(1)}, U^{ref}_{\sigma(2)}, \cdots,
U^{ref}_{\sigma(m)}]^T$. Finding the optimal placement turns out to be a NP-hard problem for
arbitrary heat-distribution matrix and reference power vector. Appendix provides the NP-hardness
proof.

Given the hardness result, we design a heuristic algorithm for the static server placement problem
based on a greedy allocation strategy.  Algorithm \ref{alg:gsp} presents the pseudocode of our
Greedy Server Placement (GSP) heuristic.

\begin{algorithm}[t]
\caption{Greedy Server Placement (GSP)}\label{alg:gsp}
\begin{algorithmic}[1]
\small \Require{The set $\machineset = \{M_1, M_2, \cdots, M_m\}$ of $m$ servers, and the reference
power $U^{ref}_j$ of each server $M_j \in \machineset$; the set $\slotset = \{S_1, S_2, \cdots,
S_m\}$ of $m$ rack slots, and the heat distribution matrix $\mathbf{D}$. } \Ensure{A mapping
$\sigma$ from rack slots to servers.}

\State {Sort the servers in descending order of reference power, i.e., $U^{ref}_1 \ge U^{ref}_2 \ge
\cdots \ge U^{ref}_m$} \State {Initialize $T^{incr}_{l} = 0$ for all $1 \le l \le m$}  \For {each
server $M_j \in \machineset$} \State{$k^* = 0$ and $T^{incr}_{max}(k^*) = \infty$} \For {each slot
$S_k \in \slotset$} \State{$T^{incr}_{max}(k) = \max_{l=1..m}\paren{T_l^{incr} + d_{l, k}\cdot
U^{ref}_j}$} \If{$T^{incr}_{max}(k) < T^{incr}_{max}(k^*)$} \State{$T^{incr}_{max}(k^*) =
T^{incr}_{max}(k)$ and $k^* = k$} \EndIf \EndFor \State{Place server $M_j$ to slot $S_{k^*}$, i.e.,
$\sigma(k^*) = j$} \State{Update $T^{incr}_l = T^{incr}_l + d_{l, k^*}\cdot U^{ref}_j$ for all
$1\le l \le m$} \State{Update $\slotset = \slotset \backslash S_{k^*}$} \EndFor
\end{algorithmic}
\end{algorithm}

First, GSP sorts the servers in descending order of reference powers (Line 1). Since the servers
that consume more power on average will have larger contributions to the temperature increases at
all inlets, they are placed first to have more flexibility in the slot selection and so to avoid
high peak temperature. Let $T^{incr}_l$ denote the existing temperature increase at the inlet of
slot $S_l$, and it is initially set to zero for all inlets (Line 2). Let $T^{incr}_{max}(k)$ denote
the maximum temperature increase if the next server $M_j \in \machineset$ is placed in slot $S_k$,
i.e.,
\begin{equation}
T^{incr}_{max}(k) = \max_{l=1..m}\paren{T_l^{incr} + d_{l, k}\cdot U^{ref}_j} \ .
\end{equation}
Server $M_j$ will be placed in one of the remaining slots $S_{k^*} \in \slotset$ that minimizes the
maximum temperature increase, i.e., $k^* = \argmin_k T^{incr}_{max}(k)$. The temperature increase
at all inlets will then be updated and the filled slot $S_{k^*}$ will be removed from the available
set $\slotset$ (Lines 12 and 13). The algorithm iterates over all servers and terminates after the
last one is placed.

For the complexity of the algorithm, sorting and initialization takes $O(m\log m)$ time. In the
iteration, placing each server incurs $O(m^2)$ time as all remaining slots are examined to
determine the maximum temperature increase at all inlets. Therefore, the overall complexity is
$O(m^3)$. This is reasonable even for a large number of servers, since the process is performed
relatively infrequently: new placement of the servers is only necessary if there are significant
alteration to the datacenter layout or when some servers are removed and new ones are introduced.

\section{Online Job Scheduling and a Fuzzy-Based Priority Approach}\label{sec:scheduling}

Once the servers have been placed in a datacenter, they will start operation by executing the
applications or jobs. In practice, jobs are submitted by different users over time, so each job
must be assigned to a server without knowing future job arrivals. This section considers online job
scheduling under a given server placement to optimize performance and energy, and to deal with
their tradeoffs.

\subsection{Greedy Scheduling Framework}

All online scheduling algorithms described in this section fall under a Greedy Scheduling Framework
(GSF), which is evoked whenever a new job arrives or an existing job completes execution. Algorithm
\ref{alg:gsf} presents the pseudocode of this framework.

\begin{algorithm}[t]
\caption{Greedy Scheduling Framework (GSF)}\label{alg:gsf}
\begin{algorithmic}[1]
\small \Require{Job queue $\Q$, and for each job $J_i \in \Q$, the processor requirement $l_i$,
processing time $P_{i, j}$ and power consumption $U_{i, j}$; Server set $\machineset$, and for each
server $M_j \in \machineset$, the number $\bar{L}_j$ of available processors, which is initialized
to $\bar{L}_j = L_j$.} \Ensure{Assignments of the newly arrived job and the jobs in $\Q$ to the
servers in $\machineset$.} \If{a new job $J_i$ arrives} \State{$j^* = 0$ and $H_{i, j^*} = \infty$}
\For {each server $M_j \in \machineset$} \If {$\bar{L}_j \ge l_i$ \& $H_{i, j} < H_{i,
j^*}$}\State{$H_{i, j^*} = H_{i, j}$ and $j^* = j$} \EndIf \EndFor \If{$H_{i, j^*} \neq \infty$}
\State{Assign job $J_i$ to server $M_{j^*}$} \State{Update $\bar{L}_{j^*} = \bar{L}_{j^*} - l_i$}
\Else \State{Put job $J_i$ in queue $\Q$ in shortest job first order} \EndIf \ElsIf{a job $J_i$
completes execution on server $M_j$} \State{Update $\bar{L}_j = \bar{L}_j + l_i$} \For{each job
$J_k \in \Q$} \If{$\bar{L}_j \ge l_k$} \State{Assign job $J_k$ to server $M_{j}$} \State{Update
$\bar{L}_{j} = \bar{L}_{j} - l_k$} \EndIf \EndFor \EndIf
\end{algorithmic}
\end{algorithm}

The variable $H_{i, j}$ shown in the pseudocode represents the cost of assigning job $J_i$ to
server $M_j$. Specifically, $H_{i, j}$ can be a single-objective cost function of job response
time, energy consumption, etc. (see \secref{single}), or it can be a composite cost function of two
or more objectives (see \secref{multi}).

For each newly arrived job $J_i$, among the servers that have sufficiently available processors to
host it, the server with the minimum cost in terms of $H_{i, j}$ will be chosen for assigning the
job (Lines 2-9). This makes the scheduling framework greedy. If no server has enough processors to
host it, the job will be put in a waiting queue $\Q$ in Shortest Job First (SJF) order \cite{Silberschatz05_OSConcept}, which is
known to optimize the average response time (Line 12). Note that
although the processing times of the jobs are server-dependent, their relative sizes are assumed to
be consistent on different servers, i.e., a fast server is fast for all jobs. Hence, SJF can be
realized by using any server as the reference for comparing the jobs' processing times. When a job
completes execution on a server and therefore releases the occupied processors, the waiting jobs in
the queue will be tested in sequence to see if they can be assigned to this server (Lines 16-18).
Whenever a job is assigned or a running job completes execution, the number of available processors
on the server will be updated (Lines 10, 15, 19). Under this greedy scheduling framework, the
assignment of each job takes $O(m)$ time, so the overall complexity is $O(mn)$ for assigning $n$
jobs.

The next two subsections will describe heuristic algorithms that minimize different single- and
multi-objective cost functions depending on the optimization criteria.

\subsection{Single-Objective Scheduling}\label{sec:single}

Single-objective scheduling considers one optimization criterion when deciding where to assign each
job. In this subsection, we will present several single-objective scheduling heuristics. Some of
them will also be used as the base algorithms for designing the more complex multi-objective
scheduling heuristics in the next subsection.

First, the following describes some single-objective heuristics proposed in the literature
\cite{Moore05_Datacenter, Tang08_CyberPhysical}.

\begin{itemize}
\item \emph{Uniform}: Assign each job randomly to a server according to the uniform distribution.

\item \emph{MinHR}: Assign each job to a server that contributes minimally to the heat recirculation in the room. The cost function is defined as
\begin{equation}\label{eq:minHR}
H^{HR}_{i, j} = \sum_{k=1}^{m} d_{k, j} \ .
\end{equation}

\item \emph{CoolestInlet}: Assign each job to a server with the lowest temperature at its inlet. The cost function is defined as
\begin{equation}\label{eq:coolestInlet}
H^{CI}_{i, j} = T_j^{in} \ ,
\end{equation}
where $T_j^{in}$ denotes the current temperature at the inlet of server $M_j$.
\end{itemize}

Note that, in \cite{Moore05_Datacenter, Tang08_CyberPhysical}, these heuristics were applied in the
offline setting, where the information of all jobs is available to the scheduler. Here, they are
cast as online heuristics. While the aim of \emph{Uniform} is to balance the workload on all
servers, \emph{MinHR} and \emph{CoolestInlet} attempt to minimize the overall heat recirculation
and to achieve a uniform temperature distribution, respectively. However, these heuristics were
proposed for the homogeneous datacenter environments, and therefore do not consider job-specific
characteristics. The following heuristics take job-dependent information into account by minimizing
the performance, energy consumption, and temperature, respectively.

\begin{itemize}
\item \emph{Perf-Aware}: Assign job $J_i$ to a server that renders the minimum response time. The cost function is defined as
\begin{equation}\label{eq:fastest}
H^{P}_{i, j} = P_{i, j} \ ,
\end{equation}
where $P_{i, j}$ denotes the execution time of job $J_i$ on server $M_{j}$.

\item \emph{Energy-Aware}: Assign job $J_i$ to a server that incurs the minimum dynamic energy consumption due to both computing and cooling. The cost function is defined as
\begin{equation}\label{eq:greenest}
H^{E}_{i, j} = E_{total}^{dync}(\delta_{i, j} = 1) \ ,
\end{equation}
where $E_{total}^{dync}$ is the total dynamic energy defined in \eqref{totalDynamic}, and it is
evaluated based on the currently running jobs and  with job $J_i$ assigned to server $M_j$, i.e.,
$\delta_{i, j} = 1$.

\item \emph{Thermal-Aware}: Assign job $J_i$ to a server that minimizes the maximum inlet temperature. The cost function is defined as
\begin{equation}\label{eq:coolest}
H^{T}_{i, j} = \max_{k=1\cdots m} \paren{T_k^{in} + \sum_{k=1}^{m}d_{k, j} \cdot U_{i, j}},
\end{equation}
where $T_k^{in}$ denotes the current temperature at the inlet of server $M_k$, and $U_{i, j}$
denotes the power consumption of job $J_i$ on server $M_j$.
\end{itemize}

Except for \emph{Uniform}, all heuristics above break the tie by randomly selecting a server with
the best cost function. The difference between \emph{CoolestInlet} and \emph{Thermal-Aware} is that
the former considers the current inlet temperature before the job is assigned, whereas the latter
considers the resulting temperature if the job is assigned to the server. Note that all of these
heuristics make greedy decisions locally for each arriving job, so they are not guaranteed to
provide the optimal global cost.

\subsection{Multi-Objective Scheduling with Fuzzy-Based Priority}\label{sec:multi}

Scheduling jobs to optimize two or more objectives usually require exploring the tradeoff between
the conflicting goals. In this subsection, we propose a novel \emph{fuzzy-based priority} approach
to handle such a tradeoff.

\paragraph{Fuzzy-Based Priority for Bi-Objective Scheduling}

We first consider optimizing two objectives, for which we define the following composite cost
function:
\begin{equation}
H_{i, j}^{X, Y} = \langle \bar{H}^X_{i, j}(f), H_{i, j}^Y\rangle \ .
\end{equation}
In this case, the objectives $X$ and $Y$ are considered one after another by first selecting all
servers that offer the best performance in terms of $X$, and then selecting among this subset any
server that offers the best performance in terms of $Y$. To avoid depriving the second objective
altogether, a \emph{fuzzy factor} $f$, where $f\in [0, 1]$, is used to relax the selection
criterion for the first objective up to a predefined margin (in percentage). The purpose is to
explore any potential improvement for $Y$ while maintaining the performance for $X$ within a
user-defined range of acceptance. The approach will be particularly effective if a small compromise
in $X$ can lead to a large improvement in $Y$. Setting $f=0$ indicates the high importance of $X$
that should not be compromised at all, while setting $f=1$ suggests that $X$ does not matter in the
optimization. Varying $f$ in between gives the user a flexible and intuitive way to specify the
tradeoff between the two objectives.

To implement the fuzzy-based priority approach in the online Greedy Scheduling Framework (GSF) as
shown in Algorithm \ref{alg:gsf}, the cost function for the first objective $X$ needs to be
normalized between 0 and 1 in order to take the fuzzy factor into account, i.e.,
\begin{equation}\label{eq:normalization}
\bar{H}^X_{i, j} = \frac{H^X_{i, j} - {H}^X_{i, min}}{H^X_{i, max} - {H}^X_{i, min}} \ ,
\end{equation}
where ${H}^X_{i, min}$ and ${H}^X_{i, max}$ denote the minimum and maximum costs in terms of
objective $X$ among all available servers to assign job $J_i$, and they can be easily collected by a linear scan of the available servers. The implementation then relies on
the following rule for comparing the relative cost of assignment on any two servers.

\textbf{\emph{Fuzzy-Based Priority Rule (for Two Objectives)}}: The costs incurred by assigning job
$J_i$ to any two servers $M_{j_1}$ and $M_{j_2}$ satisfy $H_{i, j_1}^{X, Y} < H_{i, j_2}^{X, Y}$ if
and only if one of the following conditions holds:
\begin{itemize}
\item $\bar{H}^X_{i, j_1} \le f < \bar{H}^X_{i, j_2}$, or
\vspace{0.2em}
\item $\bar{H}^X_{i, j_1} \le f$ and $\bar{H}^X_{i, j_2} \le f$ and $H^Y_{i, j_1} < H^Y_{i, j_2}$, or
\vspace{0.2em}
\item $\bar{H}^X_{i, j_1} < \bar{H}^X_{i, j_2} \le f$ and $H^Y_{i, j_1} = H^Y_{i, j_2}$, or
\vspace{0.2em}
\item $f < \bar{H}^X_{i, j_1} < \bar{H}^X_{i, j_2}$, or
\vspace{0.2em}
\item $f < \bar{H}^X_{i, j_1} = \bar{H}^X_{i, j_2}$ and $H^Y_{i, j_1} < H^Y_{i, j_2}$.
\end{itemize}

This rule can be applied to optimize any two objectives, as long as they have well-defined cost
functions, such as the ones given in \secref{single}. The value of the fuzzy factor as well as the
priority depend on the relative importance of the two objectives to optimize, which can be
determined by the user or the system administrator.

\paragraph{Extension to Multi-Objective Scheduling} The fuzzy-based priority approach can be extended to include more than
two objectives. As in the bi-objective case, we can optimize a sequence of objectives one after
another, while using a (possibly different) fuzzy factor to specify the acceptable range for each
objective. The following illustrates this method with a composite cost function consisting of $s$
objectives:
\begin{equation}
H_{i, j}^{X_1, X_2,\cdots, X_s} = \langle \bar{H}^{X_1}_{i, j}(f_1), \bar{H}^{X_2}_{i,
j}(f_2),\cdots, H_{i, j}^{X_s}\rangle \ .
\end{equation}
In this case, the servers that are ranked among the top $f_1$ percent in terms of objective $X_1$
will be selected first. Then, within this subset, the ones that fall into the top $f_2$ percent in
terms of objective $X_2$ will be further selected. This process continues until the $(s-1)$-th
objective is considered. Finally, a server that survives the first $s-1$ rounds of selection and
has the best performance in terms of the last objective $X_s$ will be chosen as the final winner.


Again, the order of the priorities and the values of the fuzzy factors should be determined by the
relative importance of different objectives to optimize.

\paragraph{Comparison with Other Approaches}

We now comment on the similarities and differences of the fuzzy-based priority approach in
comparison with a few other multi-objective optimization approaches commonly found in the
literature. \figref{schema} illustrates the basic principles of these approaches using bi-objective
scheduling as an example. \secref{related} describes some related work on the applications of these
approaches in multi-objective scheduling.

(1). \emph{Simple priority}. This is a special case of the fuzzy-based priority approach with fuzzy
factor $f = 0$. It is usually applied in settings where strict priorities are imposed on different
objectives. This approach provides better result for the first objective, but may lead to much
worse performance for the second one. In contrast, the fuzzy-based priority approach is more
effective in settings with soft (or non-strict) priorities, especially if an objective with
slightly lower priority can be significantly improved with just a little compromise for a
high-priority objective.

(2). \emph{Pareto frontier}. This approach returns a set of \emph{nondominated}
solutions\footnote{A solution is called \emph{nondominated} if no other solution has better
performance in terms of all the objectives.} to the user instead of only one solution. It is widely
applied in offline settings to quantify the tradeoffs among different objectives. In the context of
online scheduling, however, multiple solutions are hard to maintain over time, and one of the
intermediate solutions must be selected on-the-fly in order to decide where each job should be
assigned.

(3). \emph{Constraint optimization}. This approach optimizes one objective subject to certain
constraints imposed on the other(s). It is commonly applied in environments with strict or
clearly-defined requirements, e.g., job deadline or energy budget. Instead of using an absolute
value as the constraint, the fuzzy-based priority approach specifies the constraint as a relative
threshold, i.e., fuzzy factor, in terms of percentage.

(4). \emph{Weighted sum}. This approach transforms multiple objectives into a single one by
optimizing a weighted combination. Although priorities are not explicitly specified, it uses
weights to indicate the relative importance of the objectives. As different objectives can have
different units, they are often normalized in order to be combined. However, it may not be
intuitive to set the values of the weights, e.g., for time and energy.

Compared to simple priority and constraint optimization, fuzzy-based priority is particularly
suitable for scheduling HPC applications in datacenters, where no strict constraints or priority
are normally imposed on job performance or energy consumption. Compared to weighted sum,
fuzzy-based priority provides an intuitive alternative to describing the tradeoffs while specifying
soft preference of the user on the priority of the objectives. Setting an appropriate fuzzy factor
encodes such preference in an online manner. As shown in \figref{schema}, the solution returned by
fuzzy-based priority (and other approaches) when scheduling an individual job actually lies on the
pareto frontier.

\begin{figure}[t]
\centering
    \includegraphics[width=3in]{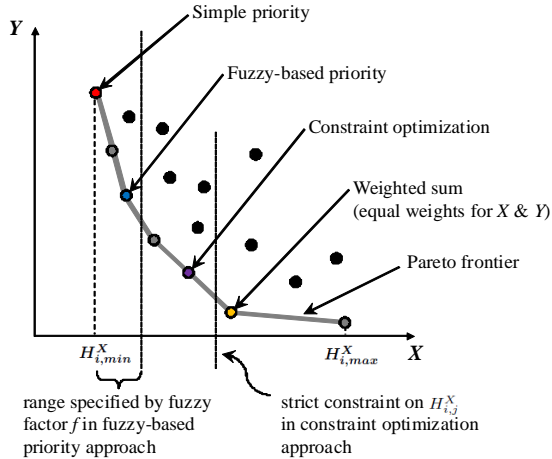}
    \caption{Comparison of the fuzzy-based priority approach with four other approaches in bi-objective scheduling.
    Each dot represents a potential solution, and the solution returned by each approach is indicated.}
    \label{fig:schema}
\end{figure}

\section{Performance Evaluations}\label{sec:evaluation}

In this section, we will evaluate the proposed online scheduling heuristics with the fuzzy-based
priority approach and the greedy heuristic for server placement. The evaluations are performed by
simulation using the Data Center Workload and Resource Management Simulator (DCworms)
\cite{Kurowskia13_DCWoRMS}.

\subsection{Simulation Setup}

\paragraph{Datacenter Configuration}

We simulate a datacenter with 50 servers and the same configuration as the one considered
in \cite{Tang06_SensorBased, Tang08_CyberPhysical}, which essentially determines the heat distribution matrix. Specifically, the datacenter consists of two rows of racks in a
typical cold aisle and hot aisle layout. The cool air is supplied by the CRAC unit from the cold
aisle between the two rows. Each row has five racks and each rack contains five servers. The server
platform used in the simulation is based on Christmann's \emph{Resource Efficient Cluster Server
(RECS)} unit \cite{Christmann09_RECS}, which is a multi-node computer system consisting of 18
processors. The datacenter consists of 900 processors in total. The RECS platform is chosen because
it represents an emerging class of high-density and energy-efficient servers with built-in power
and temperature sensors and integrated cooling support.

Table \ref{tab:parameters} shows the parameters used in the simulation, whose values are based on
real measurements in a RECS unit. From the first three parameters, the heat distribution matrix
$\mathbf{D}$ used in this paper is derived by adopting the same airflow pattern as the one measured in
\cite{Tang06_SensorBased, Tang08_CyberPhysical}. The Coefficient of Performance (CoP) is based on
the one in an HP datacenter \cite{Moore05_Datacenter} as shown by \eqref{hp_cop}.

\begin{table}[h]
\small \caption{\label{tab:parameters} Values of the parameters used in the simulation.} \centering
\begin{tabular}{| c | c |}
\hline
\emph{Parameter} & \emph{Value} \\
\hline\hline
air density ($p$) & $1.168kg/m^3$\\
\hline
air flow rate ($f_j$) & $0.1m^3/s$\\
\hline
air heat capacity ($c$) & $1004J/(^{\circ}{\rm C}\cdot kg)$\\
\hline
base power ($U^{base}_j$) & $130W$\\
\hline
redline temperature ($T^{red}$) & $25^{\circ}{\rm C}$\\
\hline
\end{tabular}
\end{table}

\paragraph{Processor Types}

To construct a heterogeneous datacenter, we select a set of five nondominated processors in terms
of performance and energy indices (the smaller the better). The performance index of a processor is
calculated as the reciprocal of its performance score measured by the passmark software
\cite{Passmark14}, which synthesizes thousands of benchmark results as the processor's performance
indicator. The energy index is simply the product of the processor's performance index and its
Thermal Design Power (TDP), which gives a relative indicator (compared to other processors) on the
average energy the processor consumes when running typical benchmarks.

\figref{processor_mix} plots the two indices for more than 500 types of processors released by
Intel between 2009 and 2013, among which five processors in the pareto frontier are selected
(marked in the figure). Table \ref{tab:processors} shows the passmark scores and TDPs of the five
selected processors. We choose these processors because they form a nondominated set, making the
scheduling problem non-trivial. In this case, no processor is dominated by others in terms of both
performance and energy consumption; hence tradeoff exists when assigning a job to different
processor types. In the simulation, each type of processor makes up 10 RECS servers with 180
computing nodes in total.

\begin{figure}[t]
\centering
    \includegraphics[width=2in]{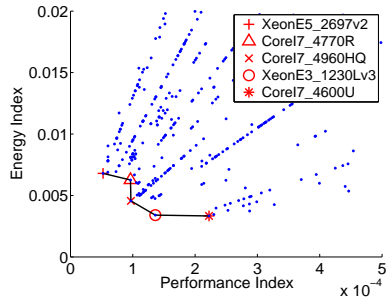}
    \caption{The performance and energy indices of 500+ processors released by Intel between 2009 and 2013. Five processors (marked) in the pareto frontier are selected for our simulation.}
    \label{fig:processor_mix}
\end{figure}

\begin{table}[h]
\small \caption{\label{tab:processors} Passmark scores (as of January 2014) and TDPs of five types
of processors used in the simulation.} \centering
\begin{tabular}{| c | c | c |}
\hline
  & Passmark & TDP(W) \\
\hline
Intel CoreI7\_4770R & 10381 & 65 \\
\hline
Intel CoreI7\_4960HQ & 10310 & 47 \\
\hline
Intel CoreI7\_4600U & 4498 & 15 \\
\hline
Intel XeonE5\_2697v2 & 19125 & 130 \\
\hline
Intel XeonE3\_1230Lv3 & 7344 & 25 \\
\hline
\end{tabular}
\end{table}

\paragraph{Benchmarks and Workloads}

The benchmarks used in the simulation consist of the following high-performance computing
applications, which are included in DCWorms.

\begin{itemize}
\item \texttt{fft}: a program to compute Fast Fourier Transforms.
\item \texttt{c-ray}: a raytracing software.
\item \texttt{abinit}: a tool to compute material properties at the atom level.
\item \texttt{linpack}: a library for performing numerical linear algebra.
\item \texttt{tar}: a program to create and manipulate tar archives.
\end{itemize}

These benchmarks exhibit a large spectrum of behaviors, from CPU intensive to memory intensive, to
communication and I/O intensive. More explanation and rationale of this choice can be found in
\cite{DaCosta13_Benchmarks}. To profile the execution time and power consumption of these
benchmarks, an application-specific approach \cite{Kurowskia13_DCWoRMS} was adopted. Specifically,
average measurements are collected for each application with different input parameters on Intel
Core I7\_2715QE, a less powerful processor available in our RECS testbed. The results are then
translated to our target platforms using the relative performance and power indicators as shown in
Table \ref{tab:processors}. \tabref{benchmarks} details the average execution time and the
corresponding power consumption of the benchmarks on each of the five selected processors.

\begin{table}[h]
\small \caption{\label{tab:benchmarks} Average execution time (above, in sec) and power consumption
(below, in Watt) of each benchmark on each type of processor.} \centering
\begin{tabular}{| c | c | c | c | c | c |}
\hline
 & CoreI7  & CoreI7  & CoreI7  & XeonE5 & XeonE3 \\
 & 4770R   & 4960HQ   & 4600U   & 2697v2  & 1230Lv3 \\
\hline \hline
\multirow{2}{*}{\texttt{fft}} & 3400 & 3450 & 7850 & 1850 & 4800\\
& 62.27 & 45.03 & 14.37 & 124.54 &  23.95\\
\hline
\multirow{2}{*}{\texttt{c-ray}} & 1150 & 1200 & 2700 & 650 & 1650 \\
& 33.70 & 24.37 & 7.78 & 67.41 & 12.96\\
\hline
\multirow{2}{*}{\texttt{abinit}} & 1700 & 1750 & 3950 & 950 & 2450 \\
& 36.11 & 26.11 & 8.33 & 72.22 & 13.89 \\
\hline
\multirow{2}{*}{\texttt{linpack}} & 3350 & 3400 & 7700 & 1850 & 4750\\
& 53.81 & 38.91 & 12.42 & 107.61 & 20.69\\
\hline
\multirow{2}{*}{\texttt{tar}} & 2000 & 2050 & 4600 & 1100 & 2800\\
& 50.92 & 36.82 & 11.75 & 101.83 & 19.58\\
\hline
\end{tabular}
\end{table}

Each job is randomly selected from one of these benchmarks and the number of processors it requires
is randomly generated from 1 to 8 with uniform distribution. Following the definition in
\cite{Downey97_Model}, the system load $\rho$ is defined to be
\begin{equation}
\rho = \frac{\lambda \cdot E[P]}{\sum_{j=1}^{m}L_j} \ ,
\end{equation}
where $\lambda$ is the arrival rate (in \#jobs per hour), $E[P]$ is the average sequential
execution time of the jobs on all processor types (roughly 4.5 hours) and ${\sum_{j=1}^{m}L_j}$ is
the total number of processors, which is 900 in the simulation. Jobs arrive according to the
Poisson process, and the arrival rate $\lambda$ is increased from 20 to 200 with a fixed arrival
duration of 8 hours. The total number of jobs ranges from 160 to 1600, and the system load is
between 0.1 and 1.

\subsection{Simulation Results}

This section presents the simulation results. First, we evaluate the performance of various online
scheduling heuristics with a fixed placement for the servers. We then study the impact of different
placement configurations on the performance of the scheduling heuristics. All results are obtained
by carrying out the experiments 10 times and taking the average.

\subsubsection{Results of Single-Objective Scheduling Heuristics}

\begin{figure*}[t]
\centering
    \subfigure[]
    {
        \label{fig:single:meanRT}
        \includegraphics[width=2in]{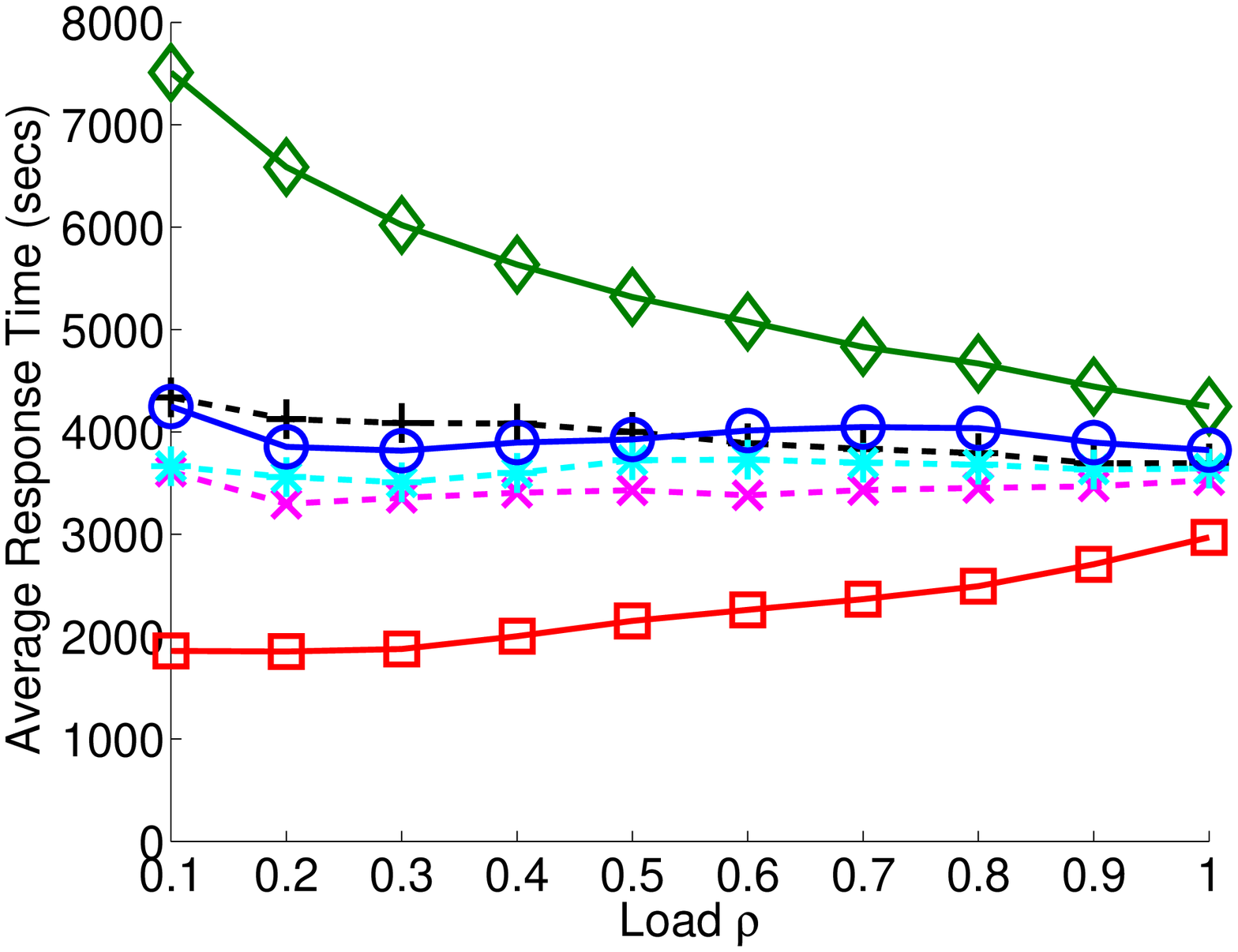}
    }
    \subfigure[]
    {
        \label{fig:single:makespan}
        \includegraphics[width=2in]{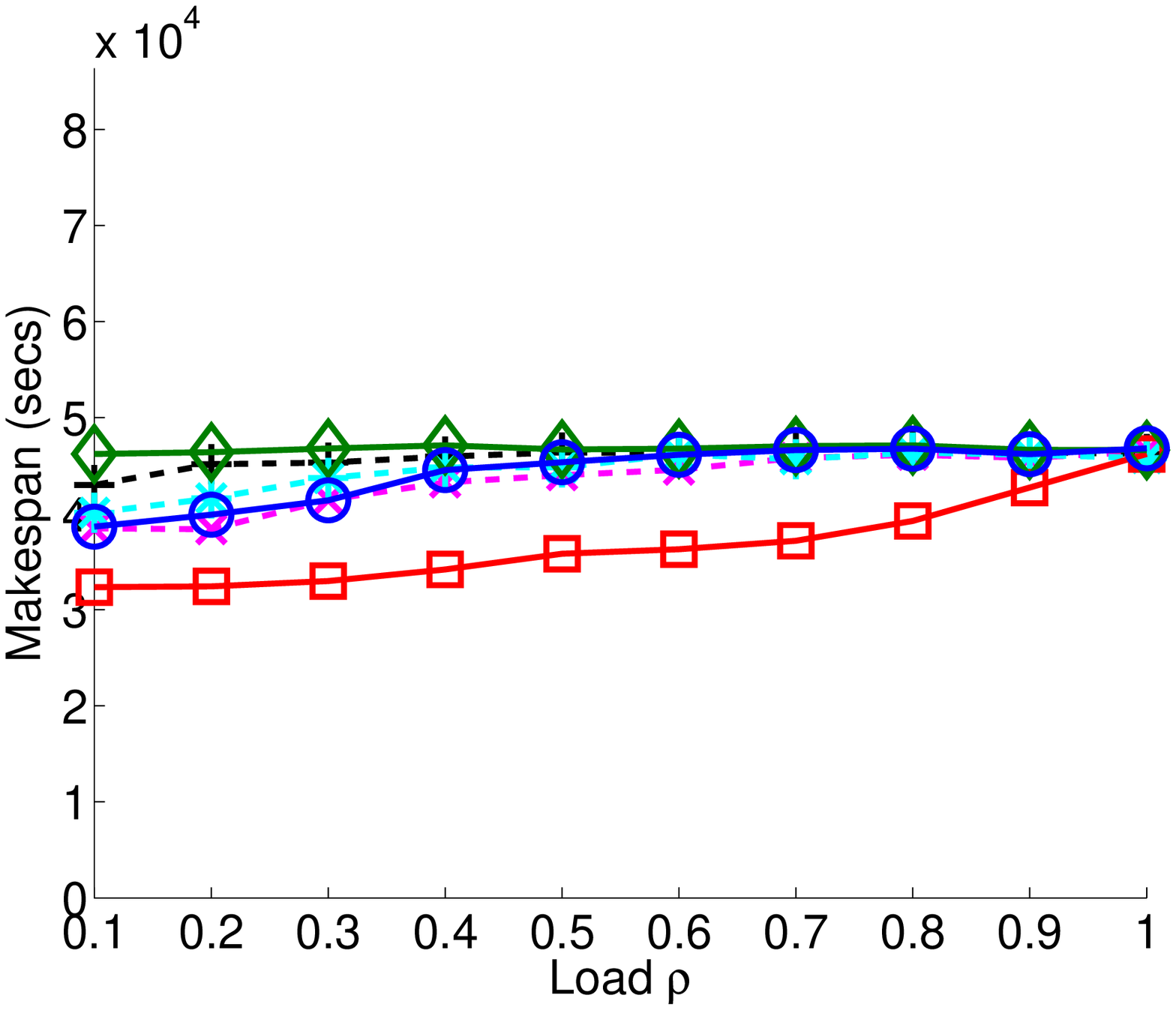}
    }
    \subfigure[]
    {
        \label{fig:single:utilization}
        \includegraphics[width=2in]{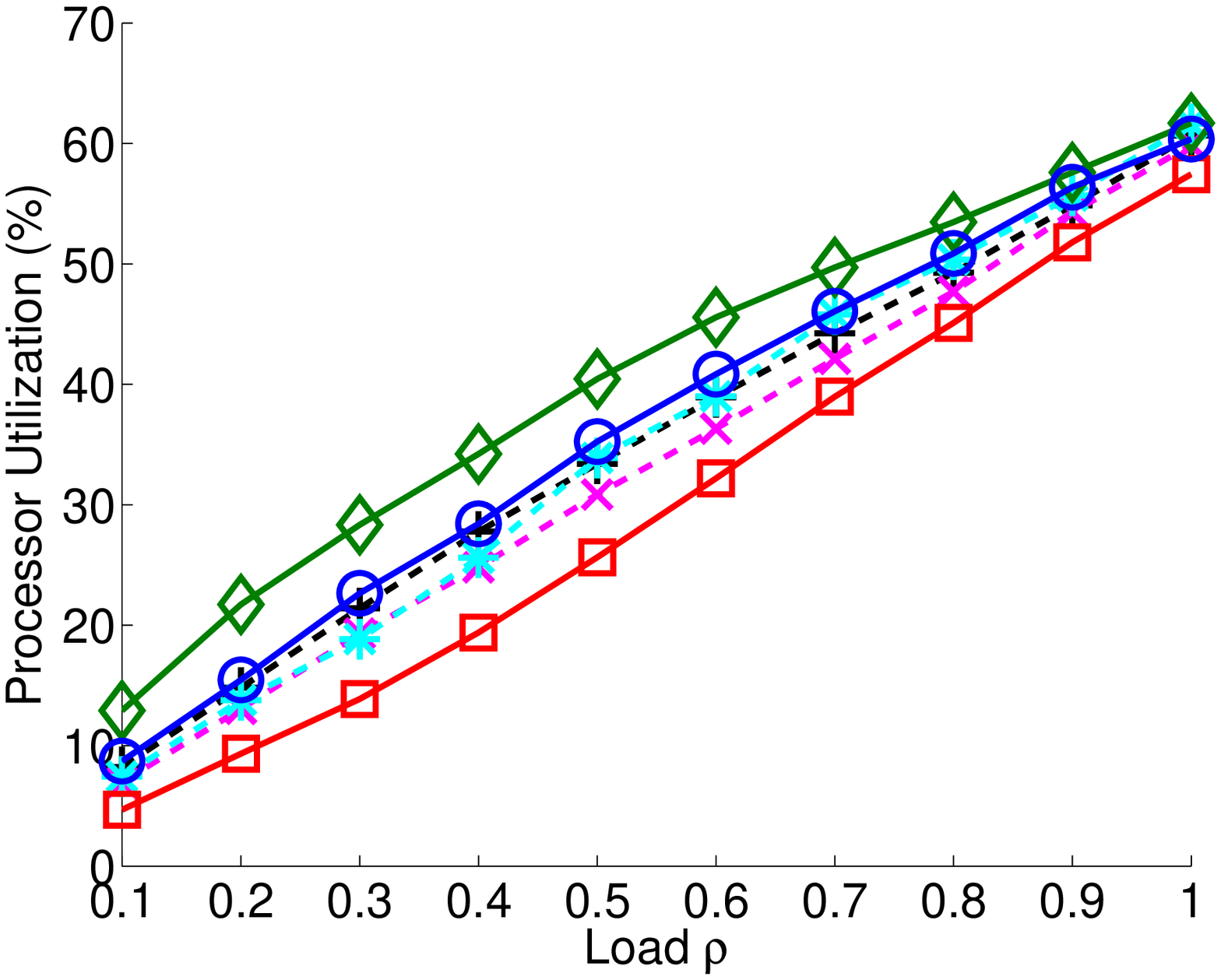}
    }
    \subfigure[]
    {
        \label{fig:single:totalEnergy}
        \includegraphics[width=2in]{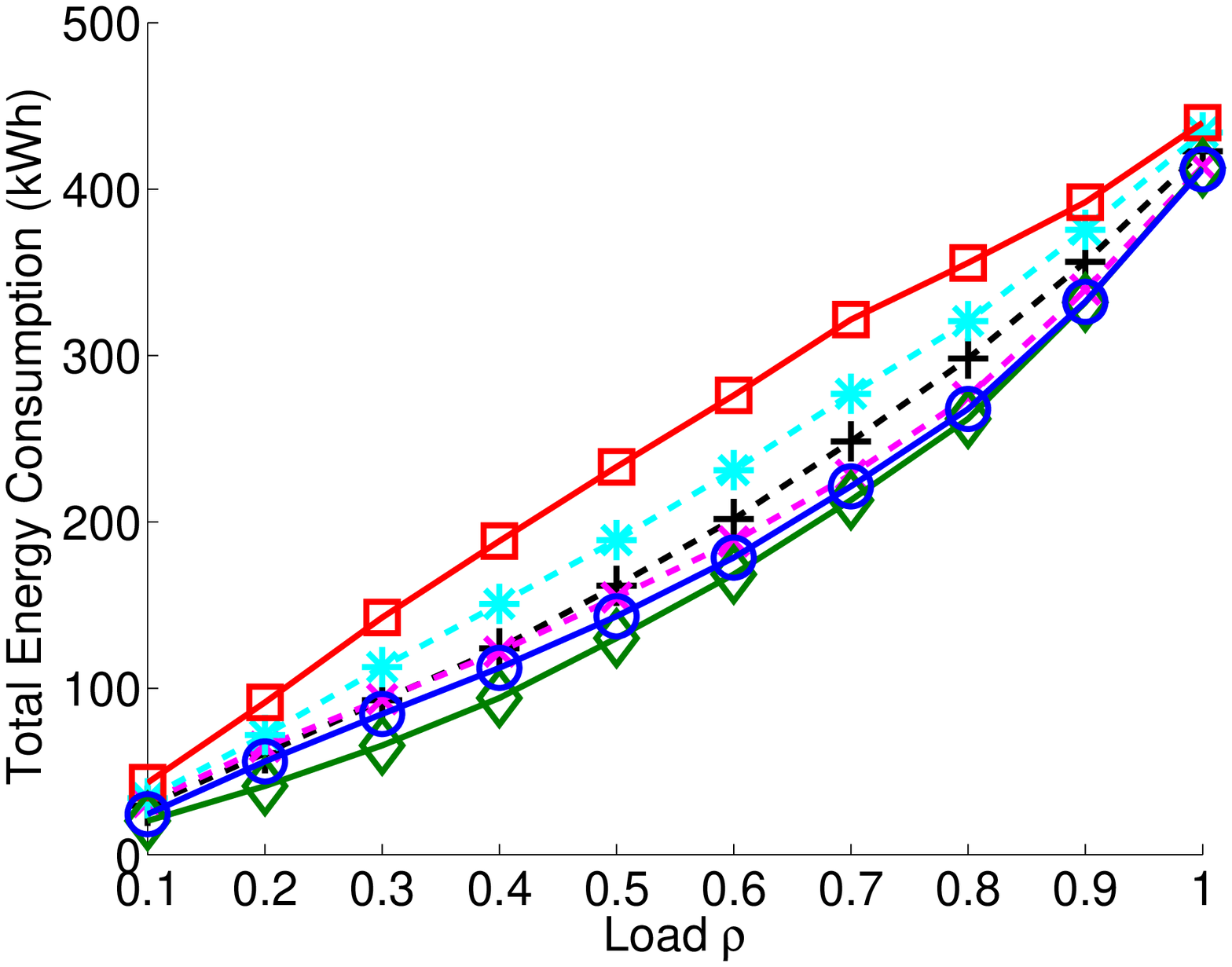}
    }
    \subfigure[]
    {
        \label{fig:single:computingEnergy}
        \includegraphics[width=2in]{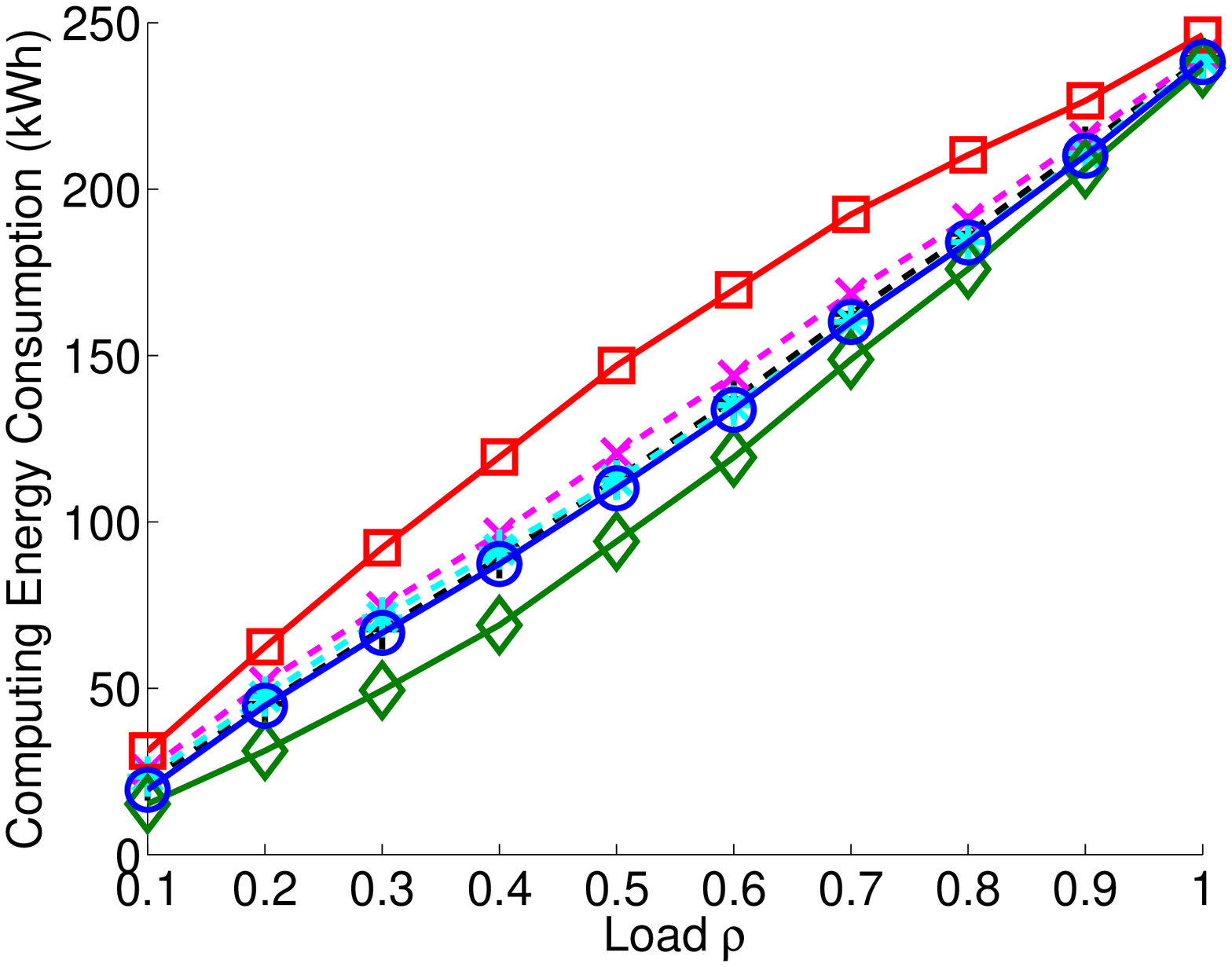}
    }
    \subfigure[]
    {
        \label{fig:single:coolingEnergy}
        \includegraphics[width=2in]{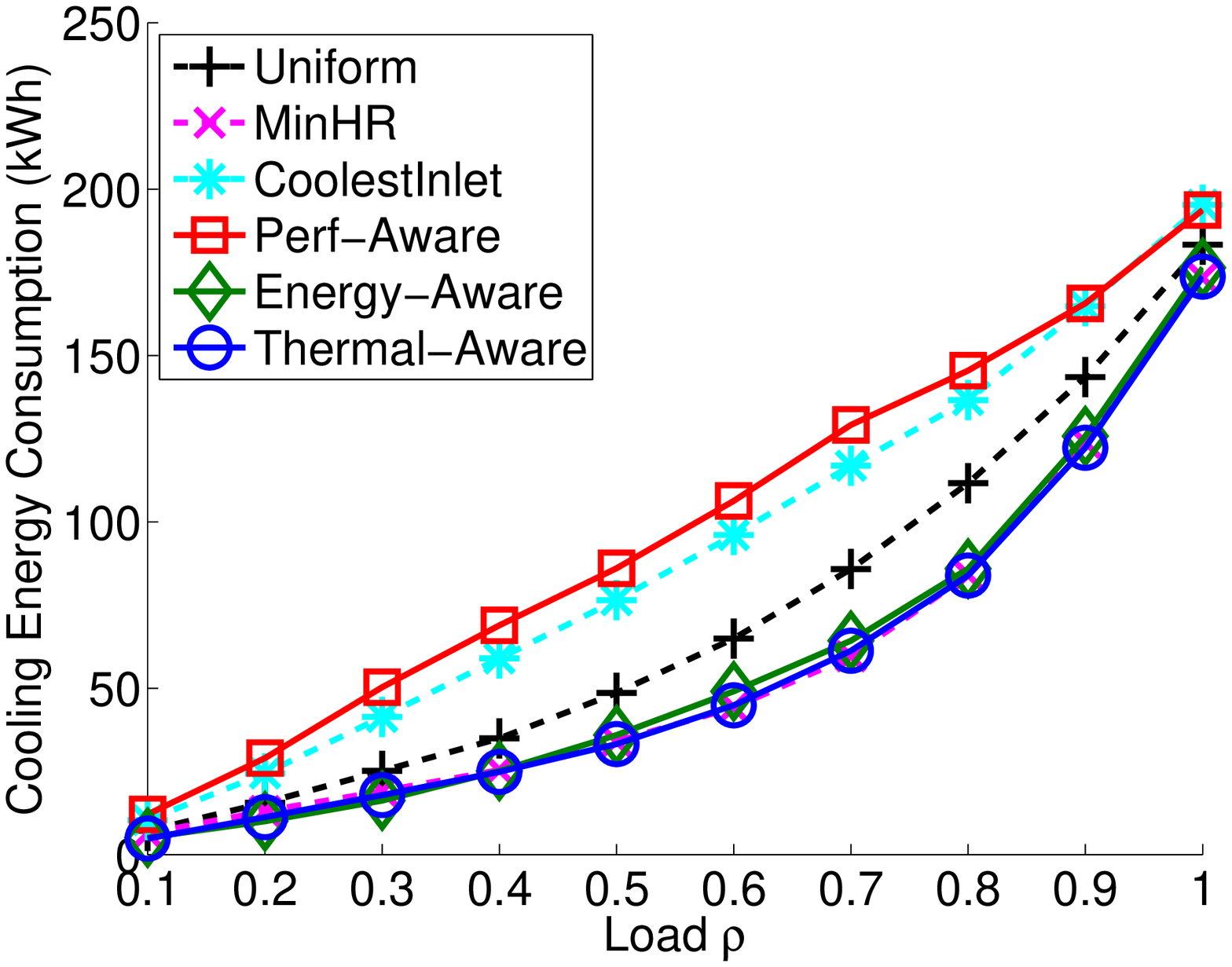}
    }
    \caption{Performance of six single-objective online scheduling heuristics. The legend applies to all subfigures.}
    \label{fig:single}
\end{figure*}

We first evaluate the online scheduling heuristics for a single objective. The results are used as
references for exploring the energy-performance tradeoff in the next subsection. In both cases, the
server placement is fixed with each type of processor occupying 10 contiguous server slots over two
racks, according to the order specified in Table \ref{tab:processors}.

Six heuristics presented in \secref{single} are evaluated, namely, \emph{Uniform}, \emph{MinHR},
\emph{CoolestInlet}, \emph{Perf-Aware}, \emph{Energy-Aware} and \emph{Thermal-Aware}.
\figref{single} presents the results of these heuristics. As we can see in \figref{single:meanRT},
\emph{Perf-Aware} has significantly better average job response time compared to the other
heuristics, especially under light system loads. This is because all jobs in \emph{Perf-Aware} are
assigned to high-performance (faster) processors before slower ones whenever possible. For the same
reason, \emph{Perf-Aware} also has better makespan (completion time of the last finished job) and
processor utilization (ratio between the utilized processor cycles and all processor cycles during
the simulation period), as shown in \figreftwo{single:makespan}{single:utilization}. Note that the
processor utilizations remain under 70\% even when the system load reaches 1. This is partly due to
the fragmented processors in some servers that cannot be utilized because a ready job simply has
higher processor requirement.

\figref{single:totalEnergy} compares the total (dynamic) energy consumption of the scheduling
heuristics, and \figreftwo{single:computingEnergy}{single:coolingEnergy} show the energy consumed
for computing and cooling, separately. For all heuristics, the energy consumption increases with
the system load or the total number of jobs in the arrival interval. \emph{Energy-Aware} consumes
less total energy compared to the other heuristics, since jobs are assigned to processors with
better energy efficiency. The improvement is more significant in terms of computing energy. For the
cooling part, \emph{MinHR} and \emph{Thermal-Aware} consumes roughly the same energy as
\emph{Energy-Aware}, since they are designed to minimize the heat recirculation and the maximum
inlet temperature, which in turn increases the supplied temperature in the room and hence directly
impacts the cooling cost. \figref{supply_temp} shows the average supply temperature of the
different scheduling heuristics in the simulation period. Indeed, \emph{Thermal-Aware} and
\emph{MinHR} are better than \emph{Energy-Aware} in terms of the average supply temperature by up
to 1.3$^{\circ}{\rm C}$ and 1.6$^{\circ}{\rm C}$, respectively.

\begin{figure}
\centering
    \includegraphics[width=2in]{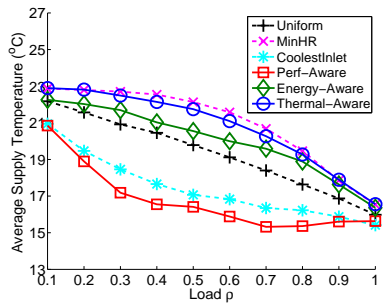}
    \caption{Average supply temperature of the heuristics.}
    \label{fig:supply_temp}
\end{figure}

\begin{figure*}[t]
\centering
    \subfigure[Load $\rho = 0.2$]
    {
        \label{fig:two1:makespan}
        \includegraphics[width=2.1in]{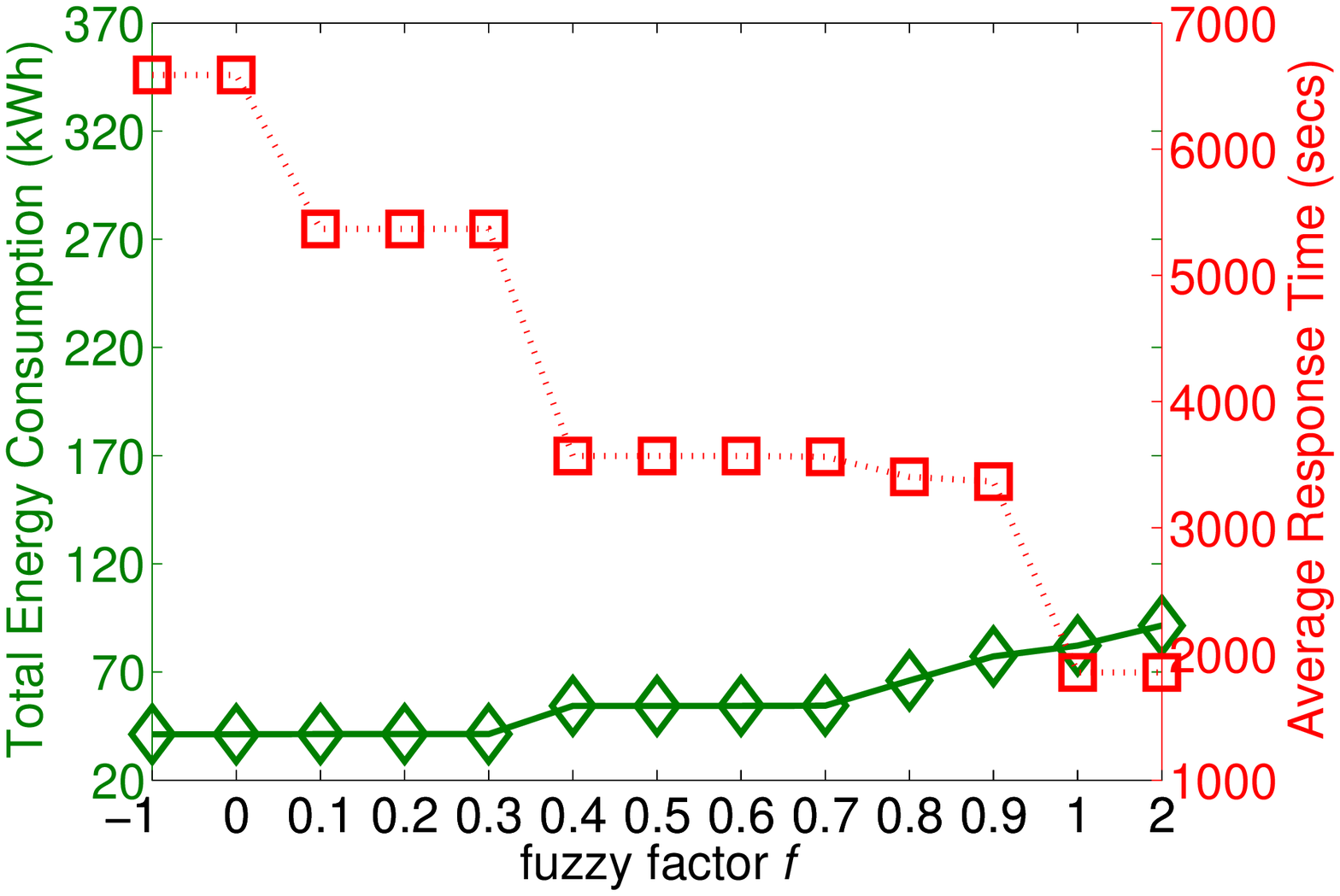}
    }
    \subfigure[Load $\rho = 0.5$]
    {
        \label{fig:two1:meanRT}
        \includegraphics[width=2.1in]{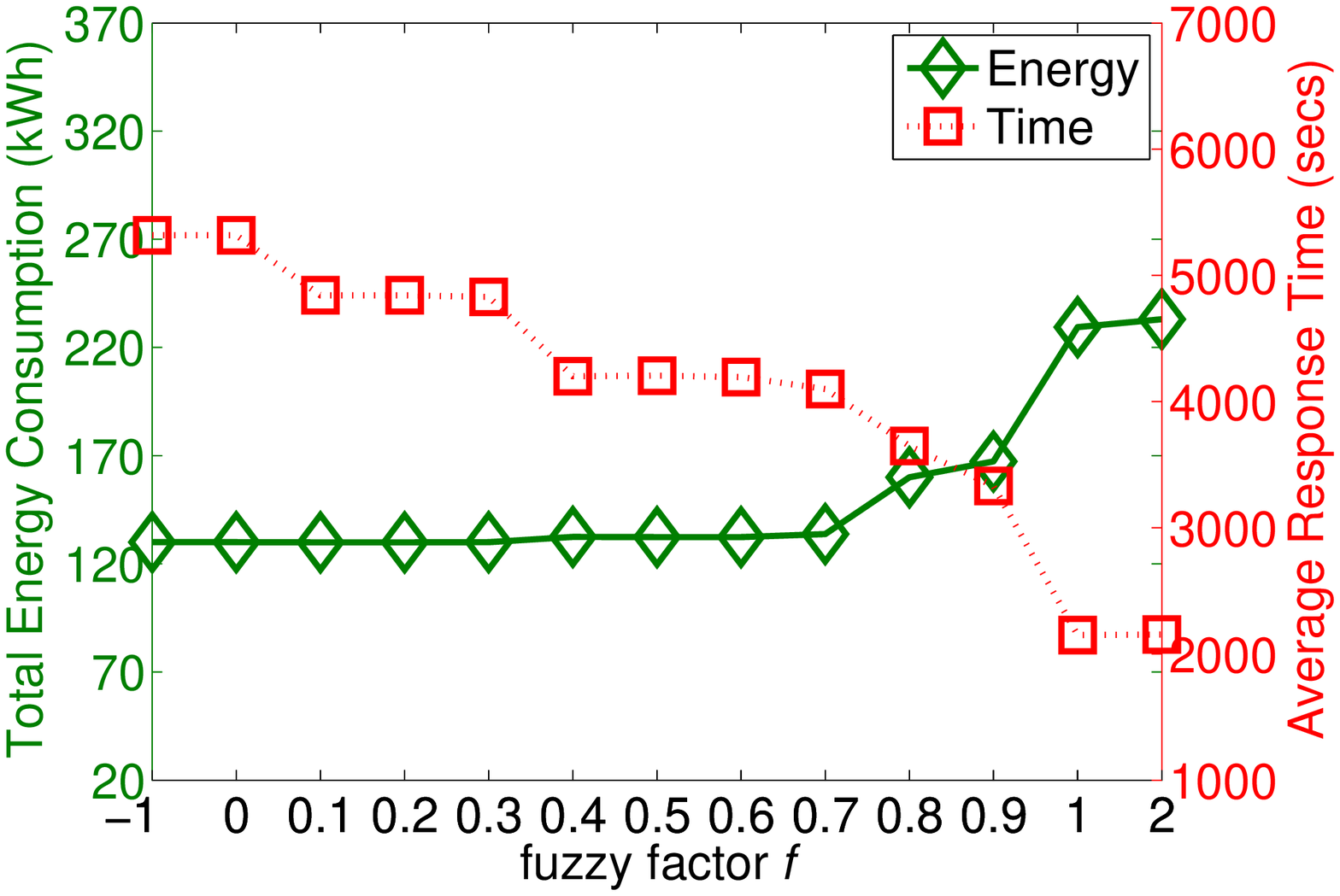}
    }
    \subfigure[Load $\rho = 0.8$]
    {
        \label{fig:two1:utilization}
        \includegraphics[width=2.1in]{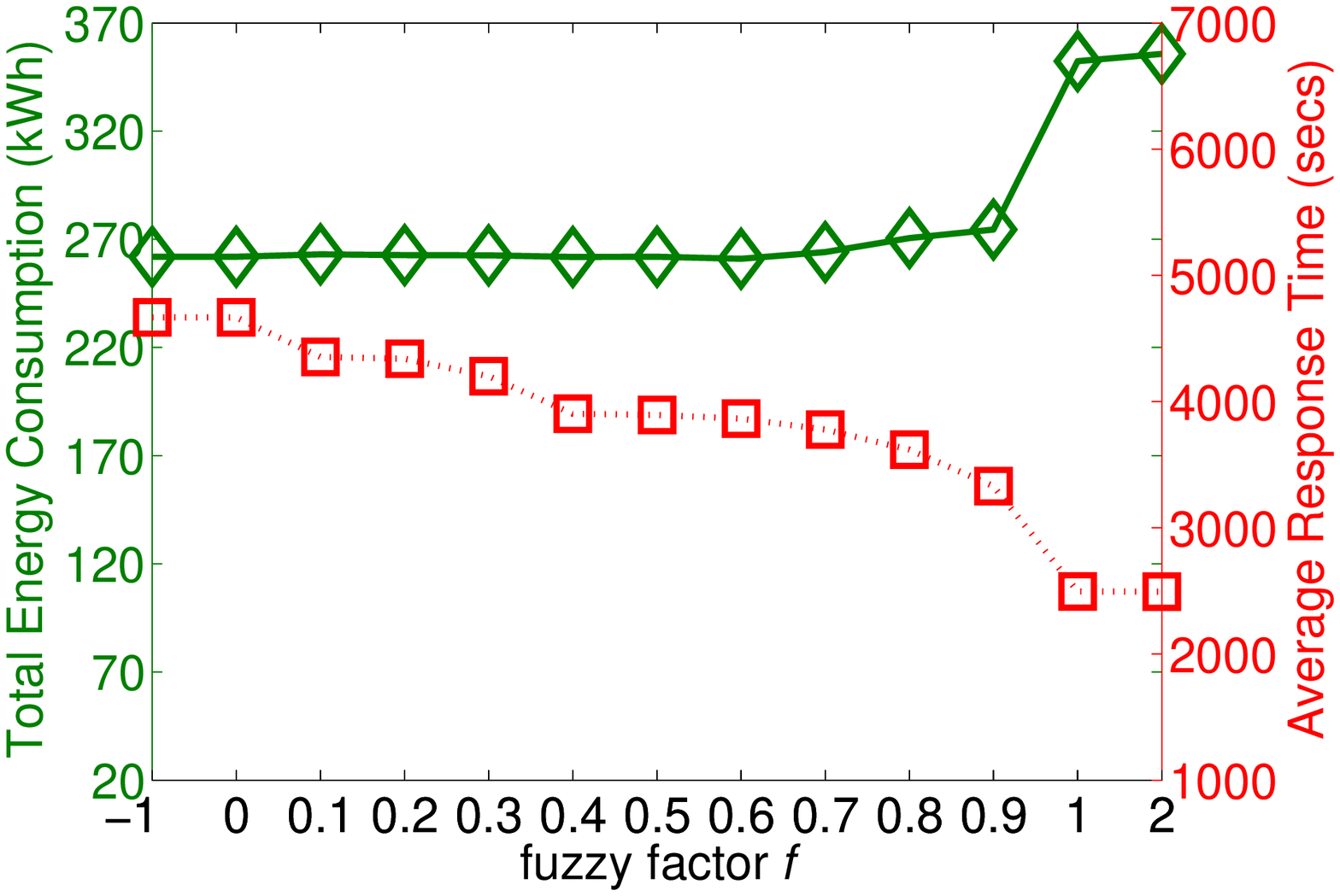}
    }
    \caption{Bi-objective scheduling for $H_{i, j}^{E, P} = \langle \overline{H}\,^{E}_{i, j}(f), H^{P}_{i, j} \rangle$ with different fuzzy factors at three system loads. The legend applies to all subfigures.}
    \label{fig:two1}
\end{figure*}

\begin{figure*}[t]
\centering
    \subfigure[Load $\rho = 0.2$]
    {
        \label{fig:two2:computingEnergy}
        \includegraphics[width=2in]{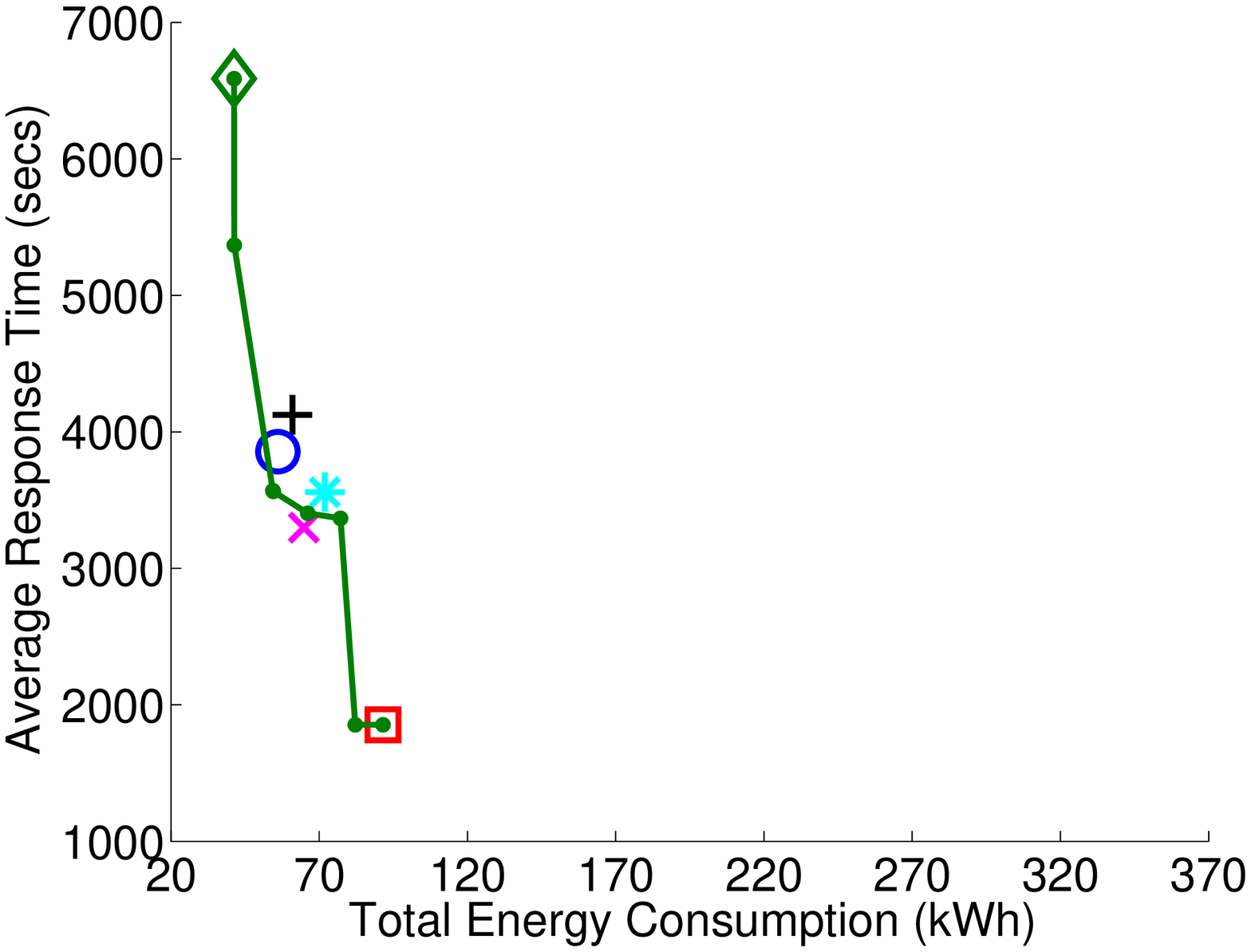}
    }
    \subfigure[Load $\rho = 0.5$]
    {
        \label{fig:two2:coolingEnergy}
        \includegraphics[width=2in]{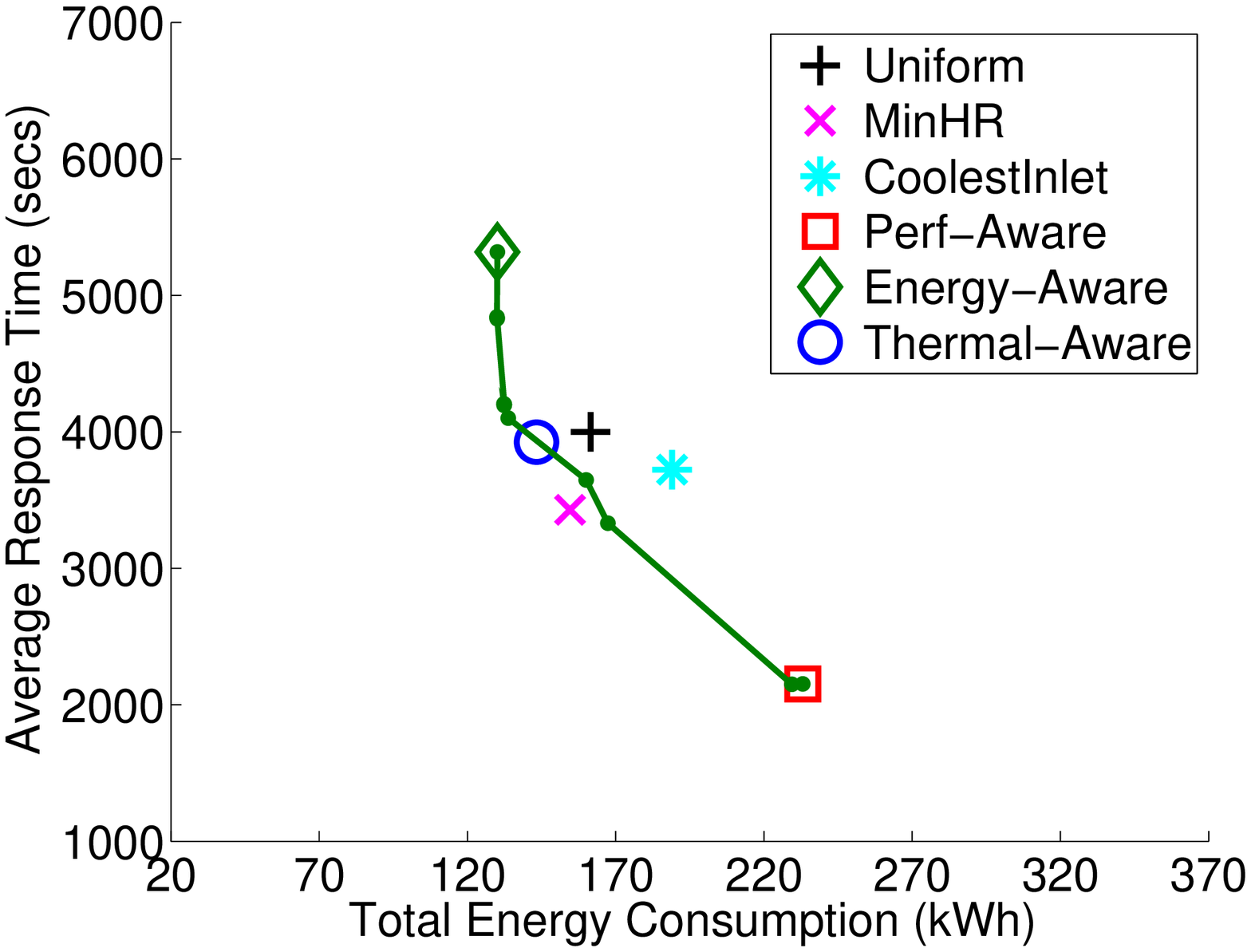}
    }
    \subfigure[Load $\rho = 0.8$]
    {
        \label{fig:two2:totalEnergy}
        \includegraphics[width=2in]{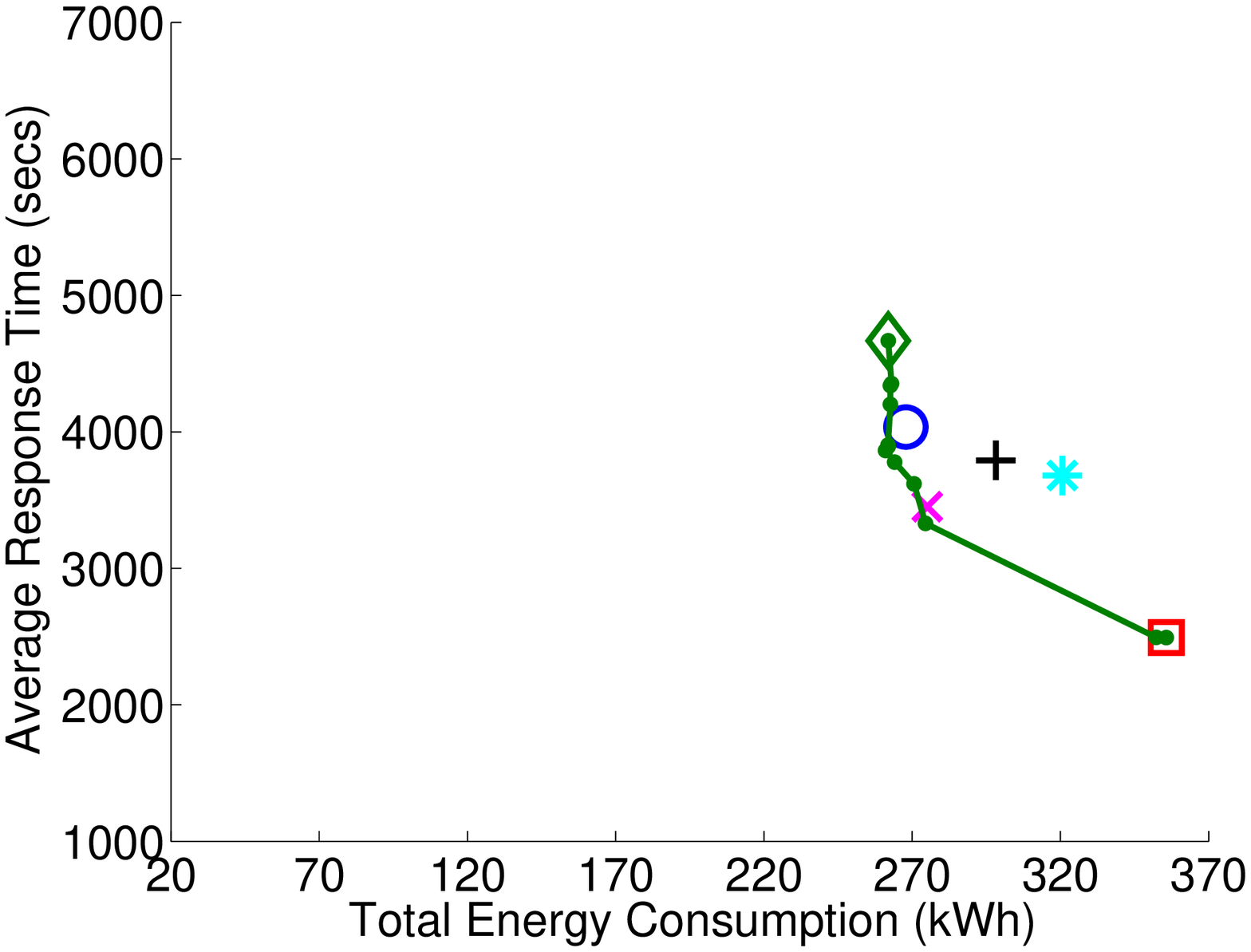}
    }
    \caption{Energy-performance tradeoff curve for $H_{i, j}^{E, P} = \langle \overline{H}\,^{E}_{i, j}(f), H^{P}_{i, j} \rangle$ at three system loads. The legend applies to all subfigures.}
    \label{fig:two2}
\end{figure*}

\begin{figure*}[t]
\centering
    \subfigure[Load $\rho = 0.2$]
    {
        \label{fig:two3:makespan}
        \includegraphics[width=2.1in]{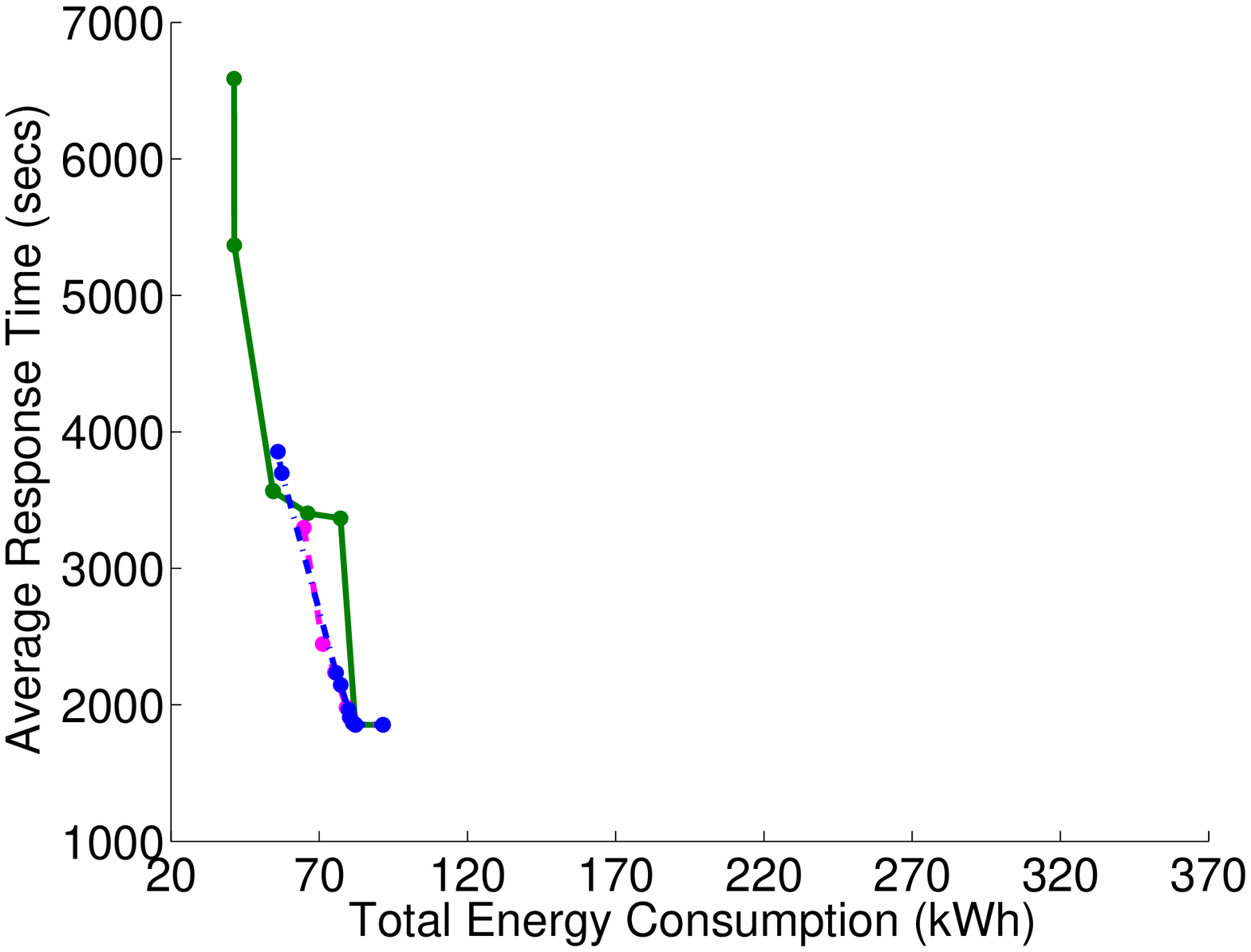}
    }
    \subfigure[Load $\rho = 0.5$]
    {
        \label{fig:two3:meanRT}
        \includegraphics[width=2.1in]{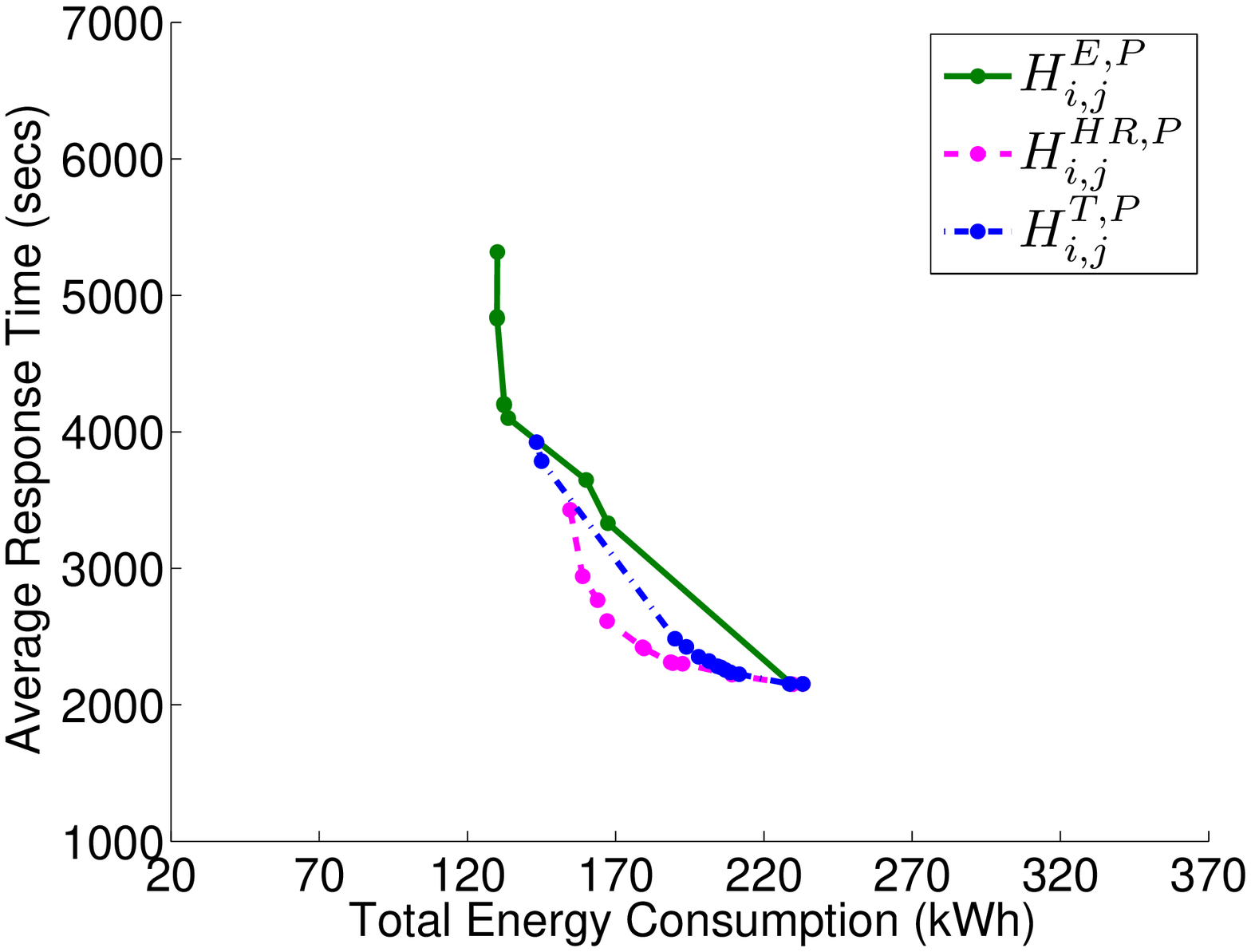}
    }
    \subfigure[Load $\rho = 0.8$]
    {
        \label{fig:two3:utilization}
        \includegraphics[width=2.1in]{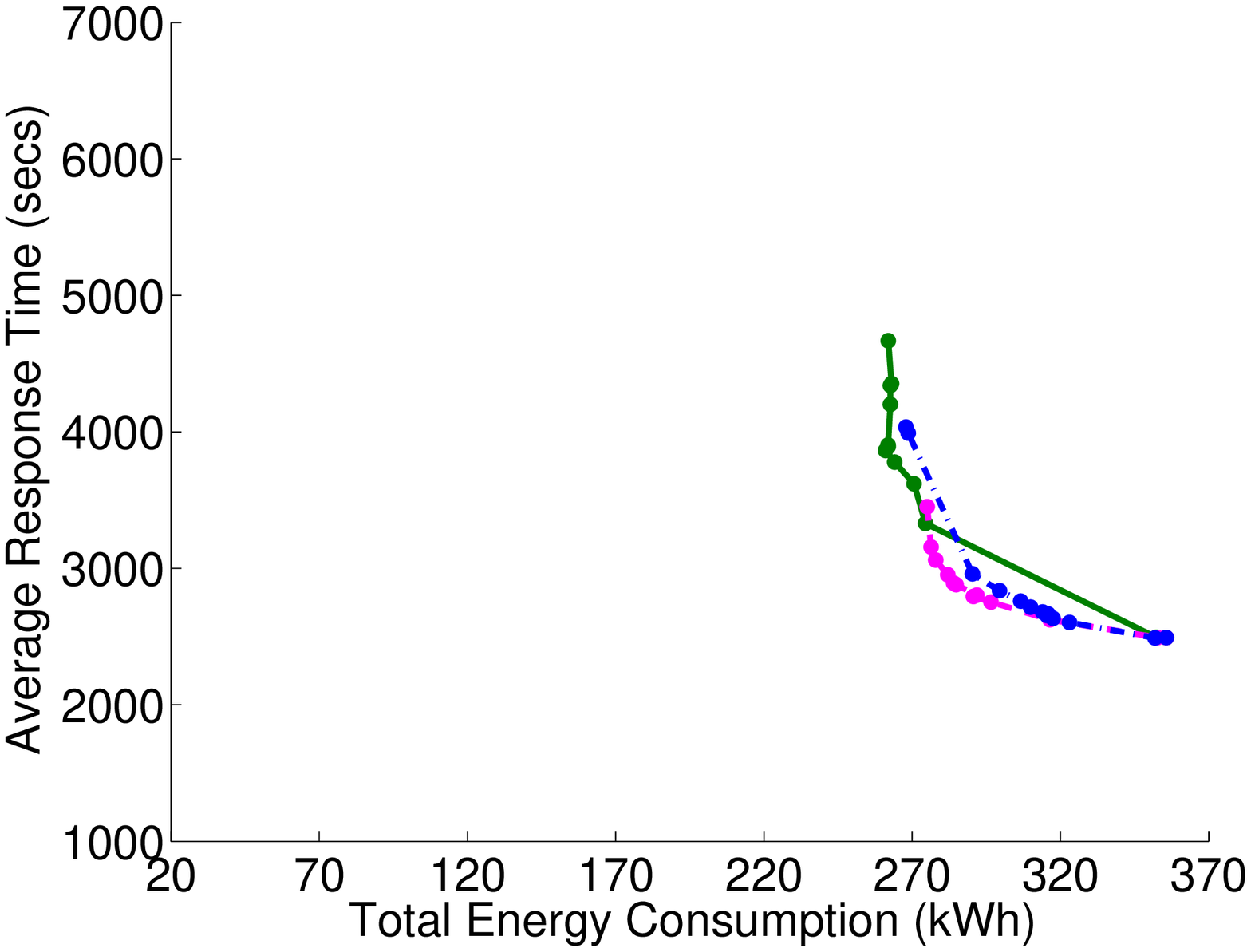}
    }
    \caption{Energy-performance tradeoff curves for $H_{i, j}^{E, P}$, $H_{i, j}^{HR, P}$ and $H_{i, j}^{T, P}$ at three system loads. The legend applies to all subfigures.}
    \label{fig:two3}
\end{figure*}

As the system load increases further and hence the processor utilization becomes higher, the
performance of all heuristics tend to converge, since all servers are roughly equally loaded under
all heuristics. In particular for \emph{Energy-Aware}, some jobs are forced to be assigned to the
high-performance servers since the energy-efficient ones are all occupied, resulting in improved
average job response time.

\subsubsection{Energy-Performance Tradeoff with Fuzzy-Based Priority}

We now evaluate the effectiveness of the fuzzy-based priority approach for exploring the
energy-performance tradeoff in online scheduling. To this end, we consider the composite cost
function $H_{i, j}^{E, P} = \langle \bar{H}^{E}_{i, j}(f), H_{i, j}^{P} \rangle$ that optimizes the
energy consumption followed by the job response time.

\figref{two1} shows the results of minimizing $H_{i, j}^{E, P}$ when the fuzzy factor $f$ is
increased from $0$ to $1$ at three different system loads ($0.2$, $0.5$ and $0.8$). The values of
both objectives are plotted as a function of $f$, with energy consumption shown on the left $Y$
axis and average response time on the right. In addition, the figure also shows the results when $f
= -1$ and $f = 2$, denoting the cases where the scheduling decision is based solely on the first
objective (energy) and the second objective (response time). The two cases are equivalent to the
single-objective heuristics \emph{Energy-Aware} and \emph{Perf-Aware}, respectively.

As we can see, the average response time improves with increased fuzzy factor at the expense of the
energy consumption under all system loads. However, the improvement can be significant even before
major compromise in energy consumption is observed. For instance, at medium load ($\rho = 0.5$),
the response time is reduced by about 1000 when $f$ reaches 0.6 without much increase in the energy
consumption. Similar results can also be observed at light load and heavy load. The fuzzy-based
priority approach can take advantage of such characteristics by setting suitable fuzzy factors in
order to achieve desirable energy-performance tradeoff in the online setting.

\figref{two2} shows the energy-performance tradeoff curve for $H_{i, j}^{E, P} = \langle
\overline{H}\,^{E}_{i, j}(f), H^{P}_{i, j} \rangle$ obtained by varying the fuzzy factor from 0 to
1. The results of the six single-objective heuristics are also shown in the figure under the
respective load. We can see that \emph{MinHR} and \emph{Thermal-Aware} lie around the curve (or
even slightly to the left of the curve in the case of \emph{MinHR}), indicating that they achieve
fairly efficient tradeoffs between job response time and energy consumption. On the other hand,
\emph{Uniform} and \emph{CoolestInlet} are completely dominated by the curve, which suggests that
they provide less attractive tradeoff results.

\figref{two3} plots the tradeoff curves achieved by optimizing the heat recirculation and the
maximum inlet temperature followed by the job response time, i.e., with cost functions $H_{i,
j}^{HR, P} = \langle \overline{H}\,^{HR}_{i, j}(f), H^{P}_{i, j} \rangle$ and $H_{i, j}^{T, P} =
\langle \overline{H}\,^{T}_{i, j}(f), H^{P}_{i, j} \rangle$. The results under three different
system loads are shown alongside the ones for $H_{i, j}^{E, P}$. The curves indicate that the two
heuristics are able to provide better tradeoffs in the medium to high energy range (e.g., between
150 and 220 for \emph{MinHR} at $\rho = 0.5$) while the tradeoff remains efficient for the cost
function $H_{i, j}^{E, P}$ when the energy consumption is close to the minimum. The results
demonstrate the flexibility of the fuzzy-based priority approach in exploring the
energy-performance tradeoff in online scheduling. The approach can be potentially applied to other
multi-objective optimization problems.

\subsubsection{Evaluation of Server Placement Strategies}\label{sec:impact_mp}

We now study the impact of server placement on the performance of the online scheduling heuristics.
Besides the simple location-based placement used in the previous evaluations, which we call LOC, we
generate three additional placements for the servers. One is based on our GSP heuristic and the
other two are based on its variations. We call the three placement configurations GSP1, GSP2 and
GSP3, respectively. The two variants (GSP2 and GSP3) are obtained in a similar fashion as GSP1. In
particular, in GSP2 the servers are sorted in ascending order of reference power instead of
descending order, and in GSP3 each server is assigned to a remaining rack slot that maximizes the
maximum inlet temperature instead of minimizing it. Apparently, these two heuristics are
counter-intuitive and are expected to provide undesirable configurations. The purpose of including
them is to demonstrate the impact of different server placements on a scheduling algorithm's
performance, especially on the cooling cost.

\figref{serverPlacement_example} shows the inlet temperature distribution of the 50 servers under
the four placement configurations. In all cases, each processor is loaded with the average power
consumption of the benchmarks shown in \tabref{benchmarks}. As we can see, GSP1 has better thermal
balance than the other configurations. Specifically, it improves LOC by about $8^{\circ}{\rm C}$ in
terms of the maximum inlet temperature and improves GSP2 and GSP3 by over $14^{\circ}{\rm C}$ and
$16^{\circ}{\rm C}$, respectively.

\begin{figure}[t]
\centering
    \includegraphics[width=3.3in]{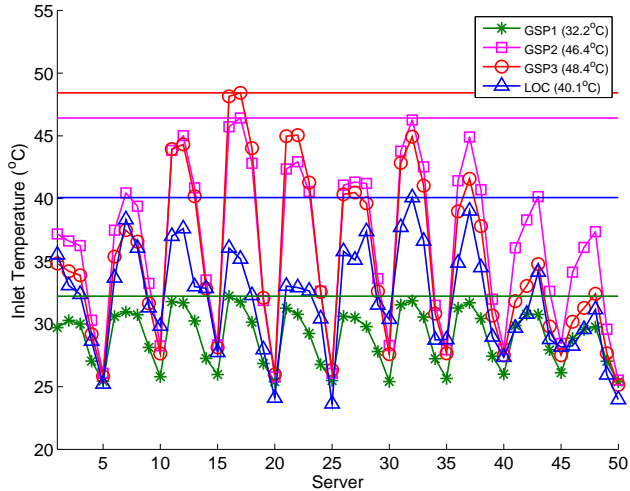}
    \caption{Inlet temperature distribution of the 50 servers under four different server placements. The maximum inlet temperature of each placement is indicated in the legend and by the horizontal line.}
    \label{fig:serverPlacement_example}
\end{figure}

\figreftwo{sp_perf}{sp_energy} show the performance of \emph{Perf-Aware} and \emph{Energy-Aware}
under the four server placements at different system loads. In both heuristics, job response time
and computing energy are not affected by different configurations. However, GSP1 has reduced
cooling energy compared to the other configurations. This is particularly evident under heavy
system load, where all servers are almost fully and equally loaded, thus their power consumption
ratios match closely those of the average values used in the server placement heuristic. Under
light system load, however, the servers could experience unbalanced loads, which causes their power
consumption ratios to deviate from those of the average values. As a result, the advantage of GSP1
becomes smaller or even diminishes, but since the overall energy consumption is small in this case,
the impact of server placement is not significant.

\begin{figure*}[t]
\centering
    \subfigure[]
    {
        \label{fig:sp_perf:meanRT}
        \includegraphics[width=2.1in]{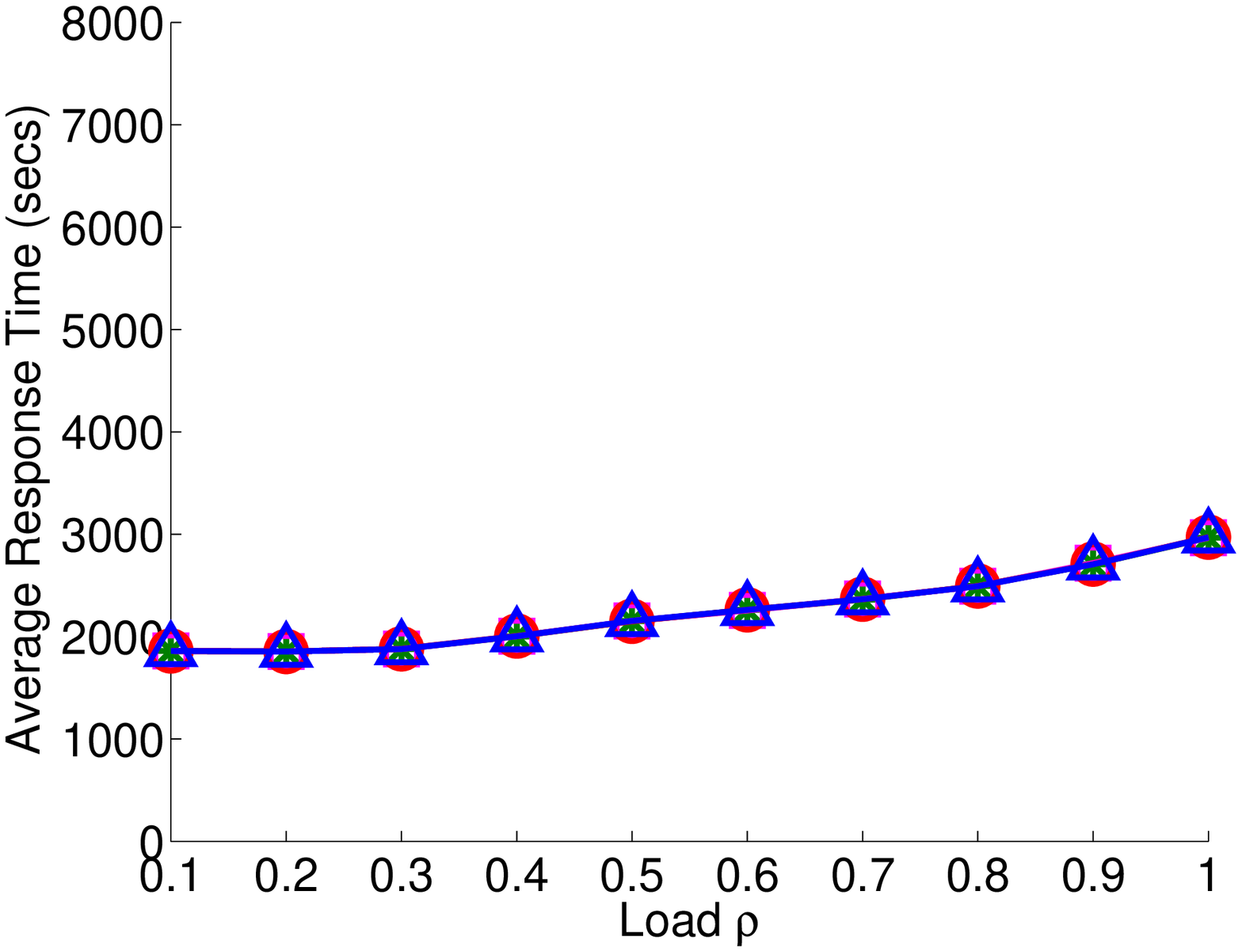}
    }
    \subfigure[]
    {
        \label{fig:sp_perf:computingEnergy}
        \includegraphics[width=2.1in]{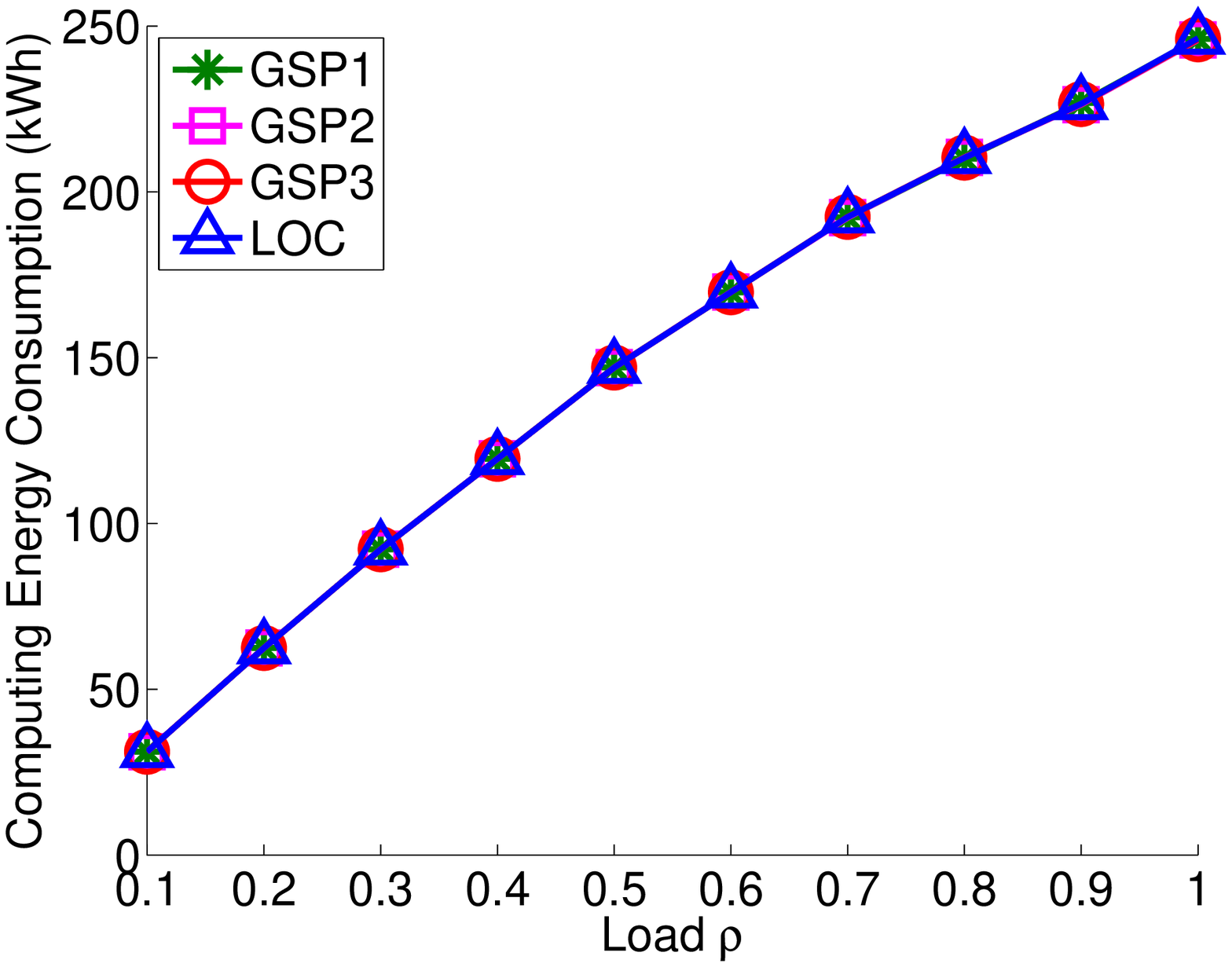}
    }
    \subfigure[]
    {
        \label{fig:sp_perf:coolingEnergy}
        \includegraphics[width=2.1in]{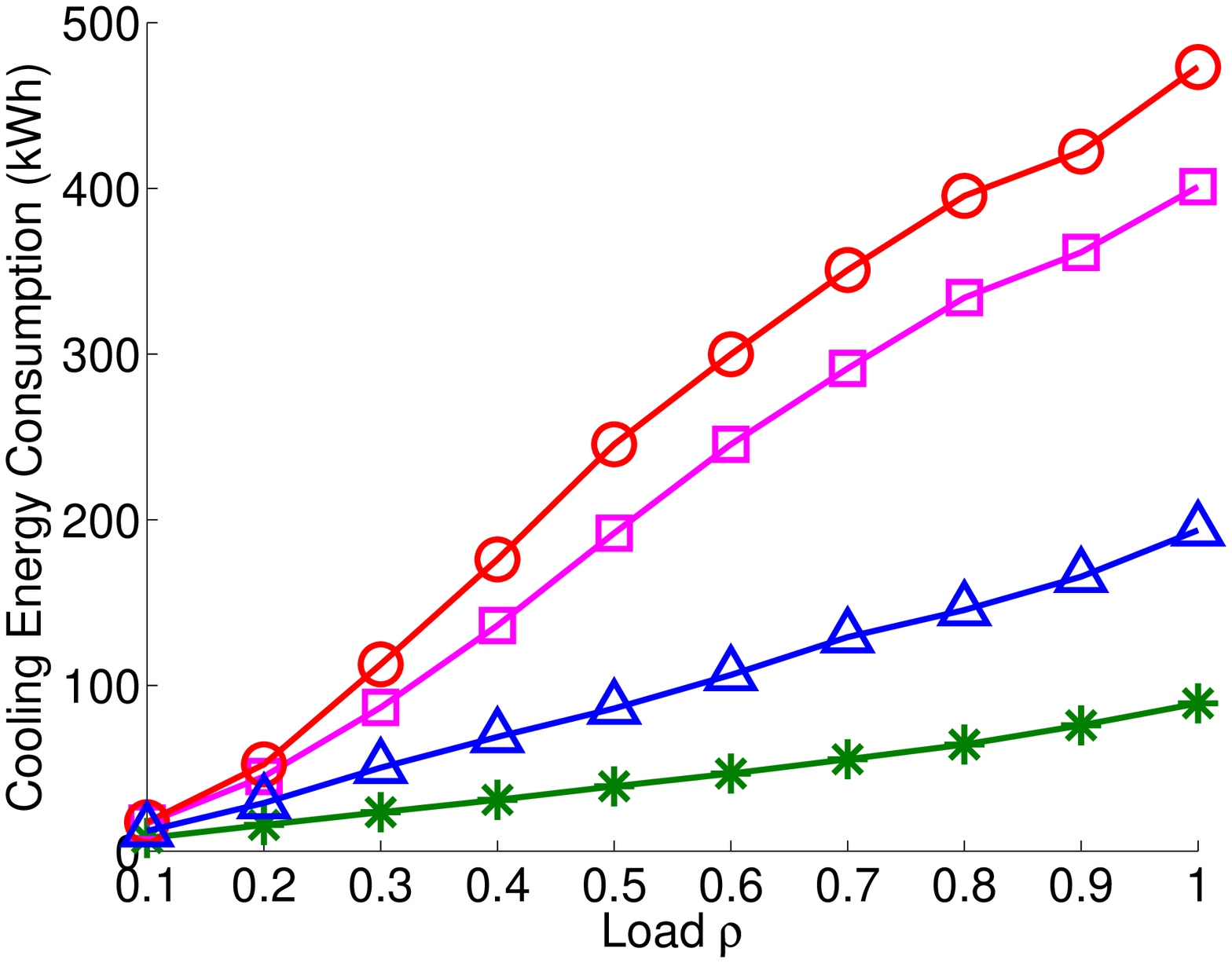}
    }
    \caption{Performance of \emph{Perf-Aware} under different server placements and system loads. The legend applies to all subfigures.}
    \label{fig:sp_perf}
\end{figure*}

\begin{figure*}[t]
\centering
    \subfigure[]
    {
        \label{fig:sp_energy:meanRT}
        \includegraphics[width=2.1in]{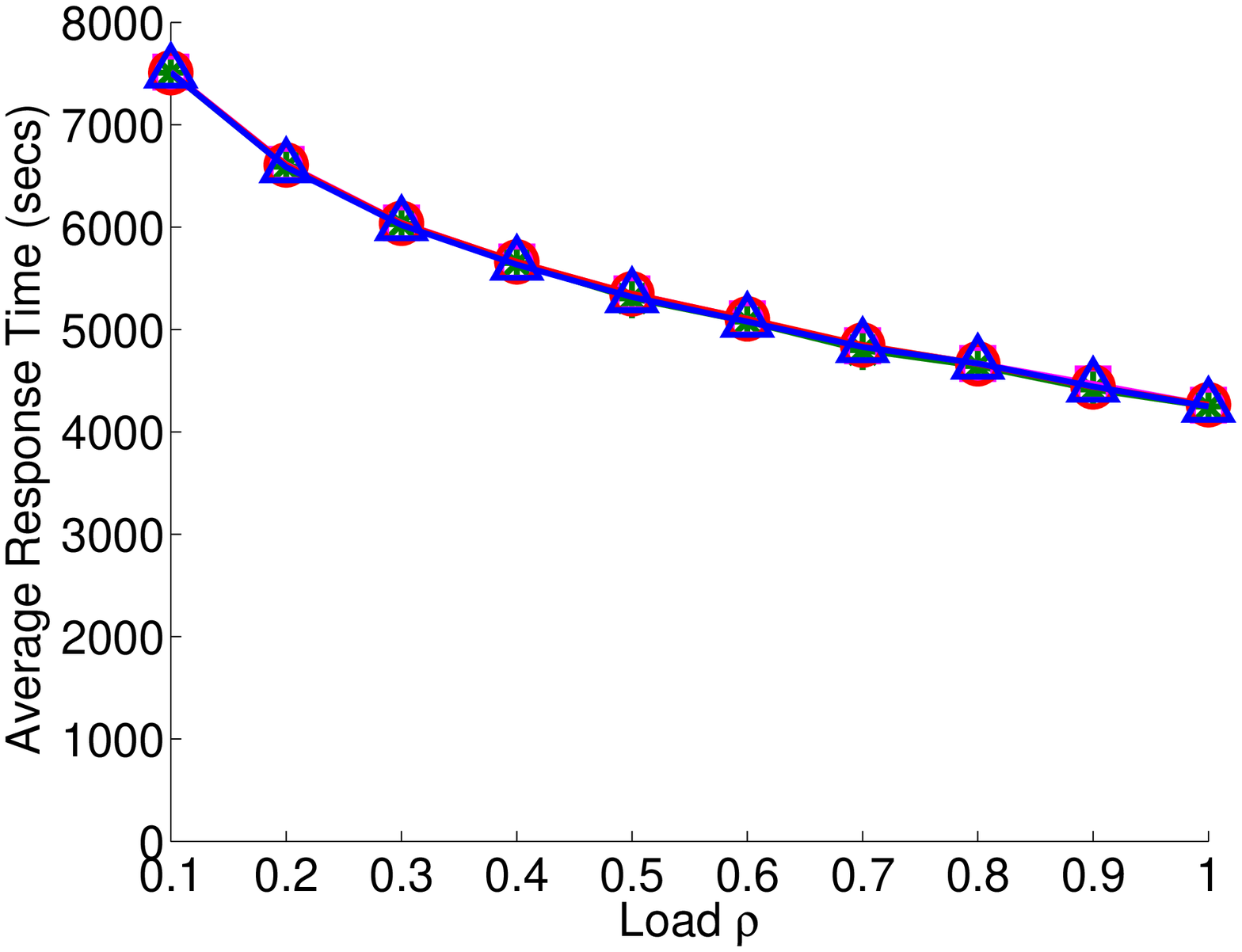}
    }
    \subfigure[]
    {
        \label{fig:sp_energy:computingEnergy}
        \includegraphics[width=2.1in]{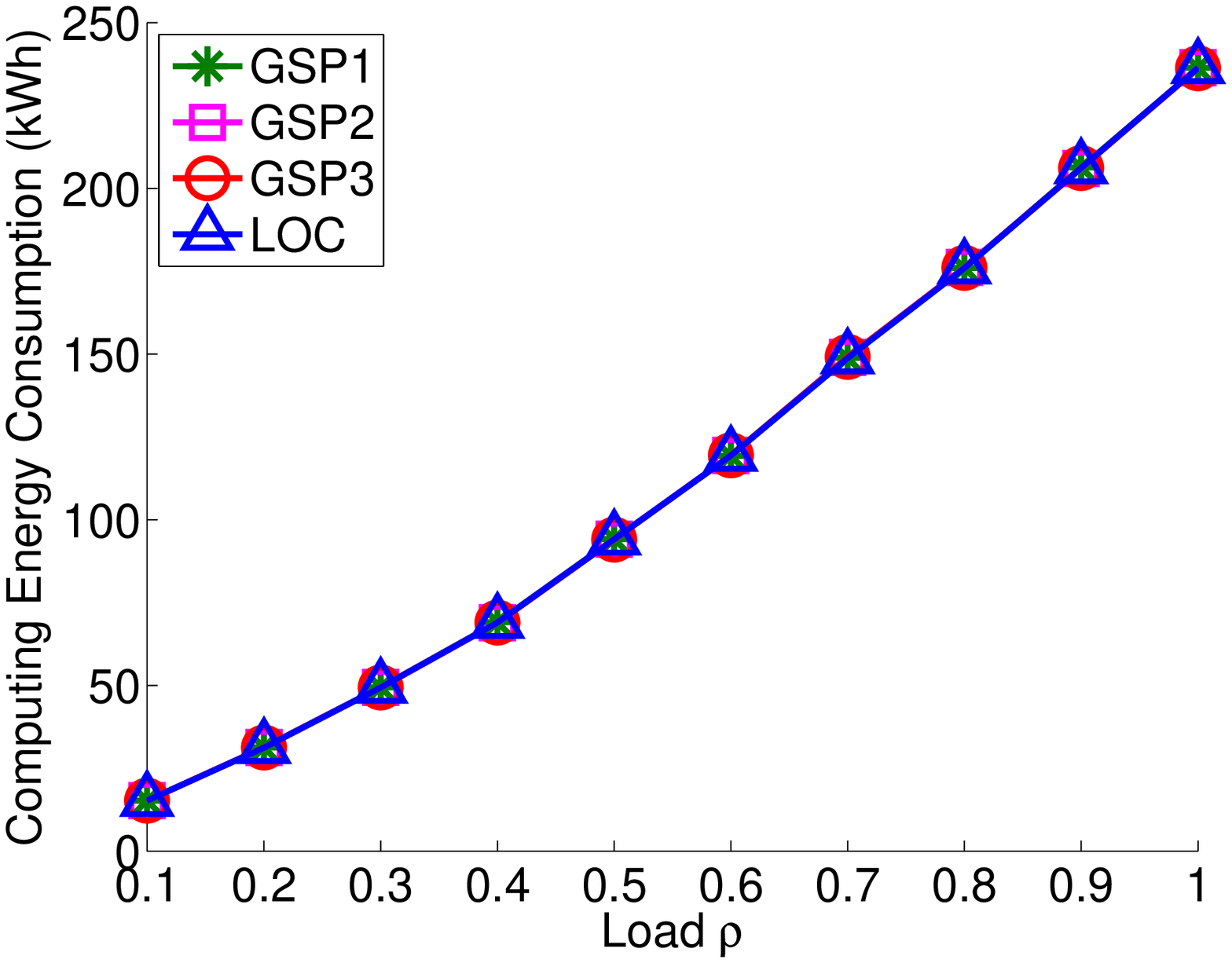}
    }
    \subfigure[]
    {
        \label{fig:sp_energy:coolingEnergy}
        \includegraphics[width=2.1in]{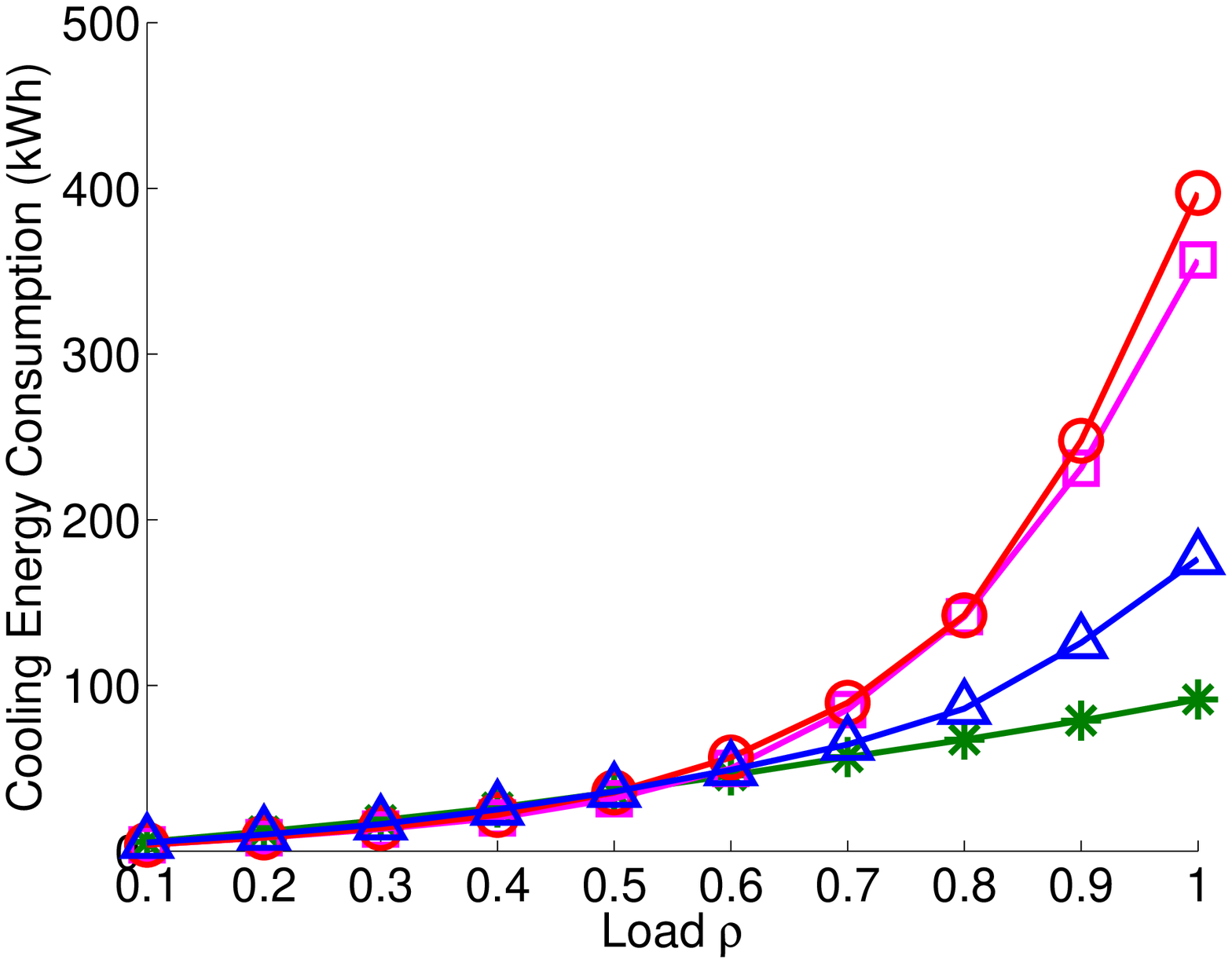}
    }
    \caption{Performance of \emph{Energy-Aware} under different server placements and system loads. The legend applies to all subfigures.}
    \label{fig:sp_energy}
\end{figure*}

Quite similar effect on the cooling energy can be observed for \emph{Thermal-Aware} and
\emph{MinHR} as shown in \figreftwo{sp_thermal}{sp_minHR}. Notice that, for these two heuristics,
different server placements also lead to a tradeoff between job response time and computing energy.
To further investigate the tradeoff efficiency, \figref{tradeoff_mix} shows the energy-performance
tradeoff curves for three heuristics with cost functions $H_{i, j}^{E, P}$, $H_{i, j}^{HR, P}$ and
$H_{i, j}^{T, P}$ at load $\rho = 0.8$ under different server placements. We can see that, although
the tradeoff remains, in all cases GSP1 provides the best cooling energy and hence improves the
overall tradeoff efficiency. Note that \emph{MinHR} and \emph{Perf-Aware} behave exactly the same
under GSP1, since servers with faster processors and hence more power consumptions are placed in
the slots with less heat recirculation. Therefore, the same performance and energy are observed for
$H_{i, j}^{HR, P}$ regardless of the fuzzy factor, as shown in \figref{tradeoff_mix:minHR}.

The results confirm that strategic server placement indeed improves the thermal balance in a
heterogeneous datacenter, which helps reduce the cooling cost. This is achieved with little impact
on the job response time and computing energy, or the tradeoff between them.

\begin{figure*}[t]
\centering
    \subfigure[]
    {
        \label{fig:sp_thermal:meanRT}
        \includegraphics[width=2.1in]{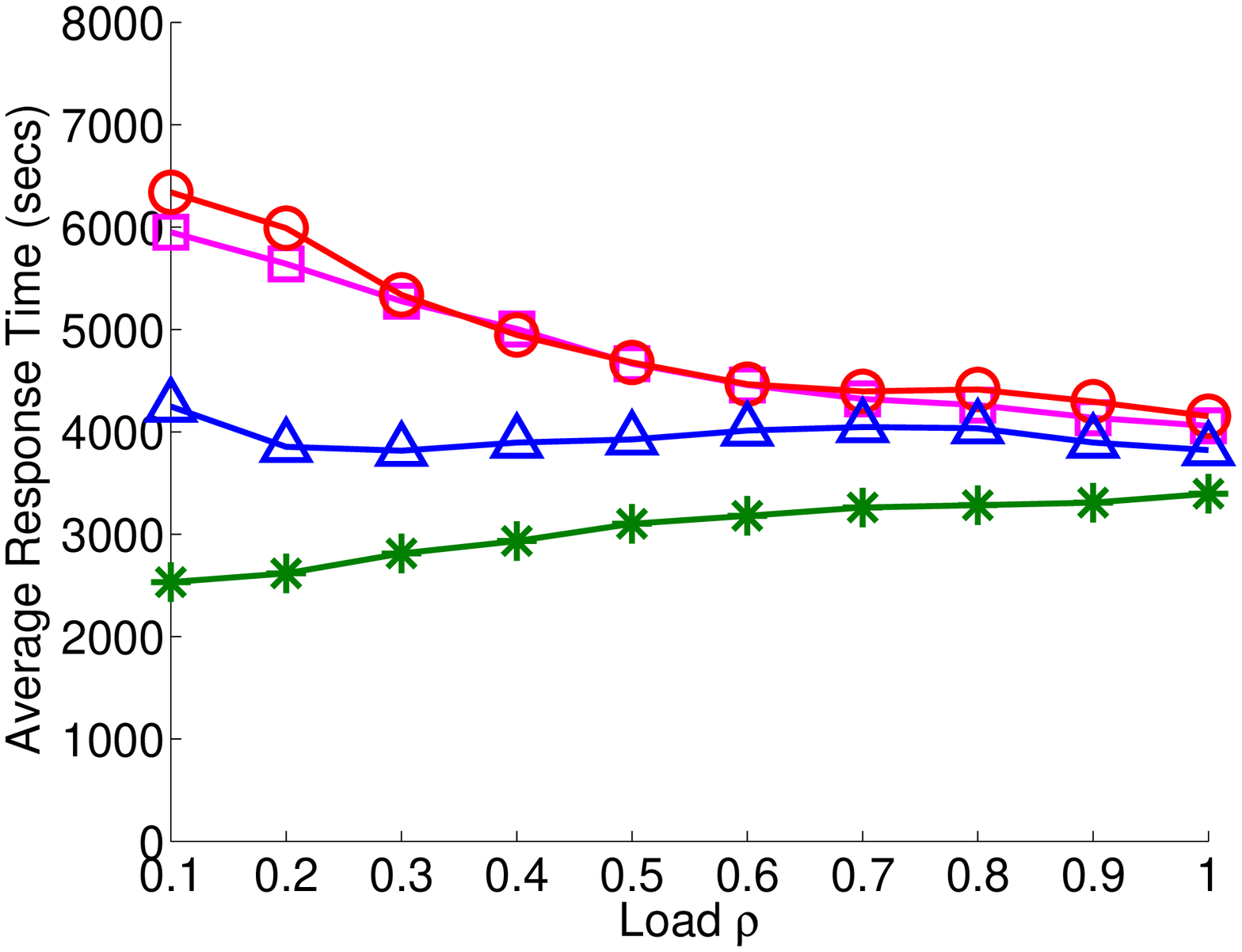}
    }
    \subfigure[]
    {
        \label{fig:sp_thermal:computingEnergy}
        \includegraphics[width=2.1in]{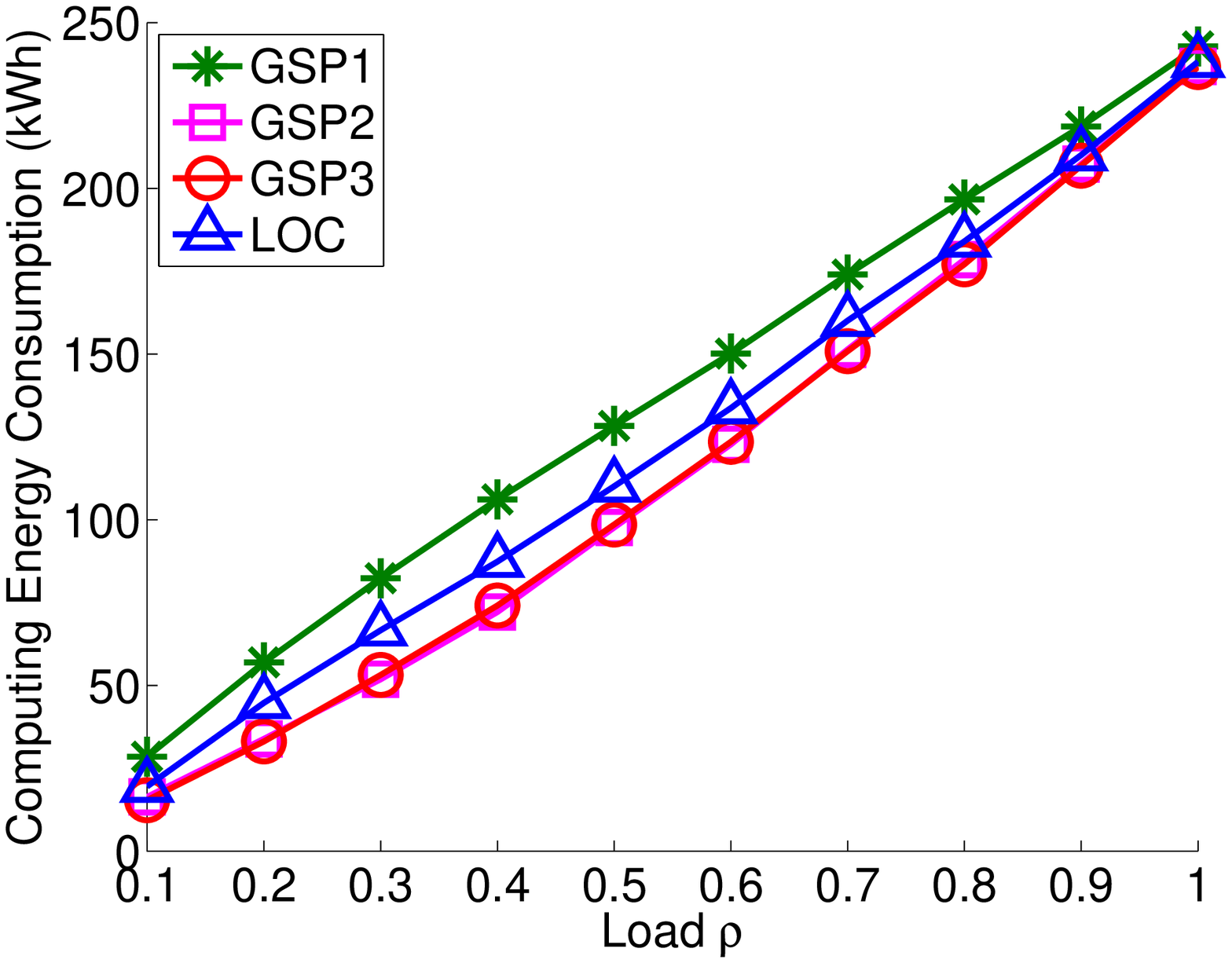}
    }
    \subfigure[]
    {
        \label{fig:sp_thermal:coolingEnergy}
        \includegraphics[width=2.1in]{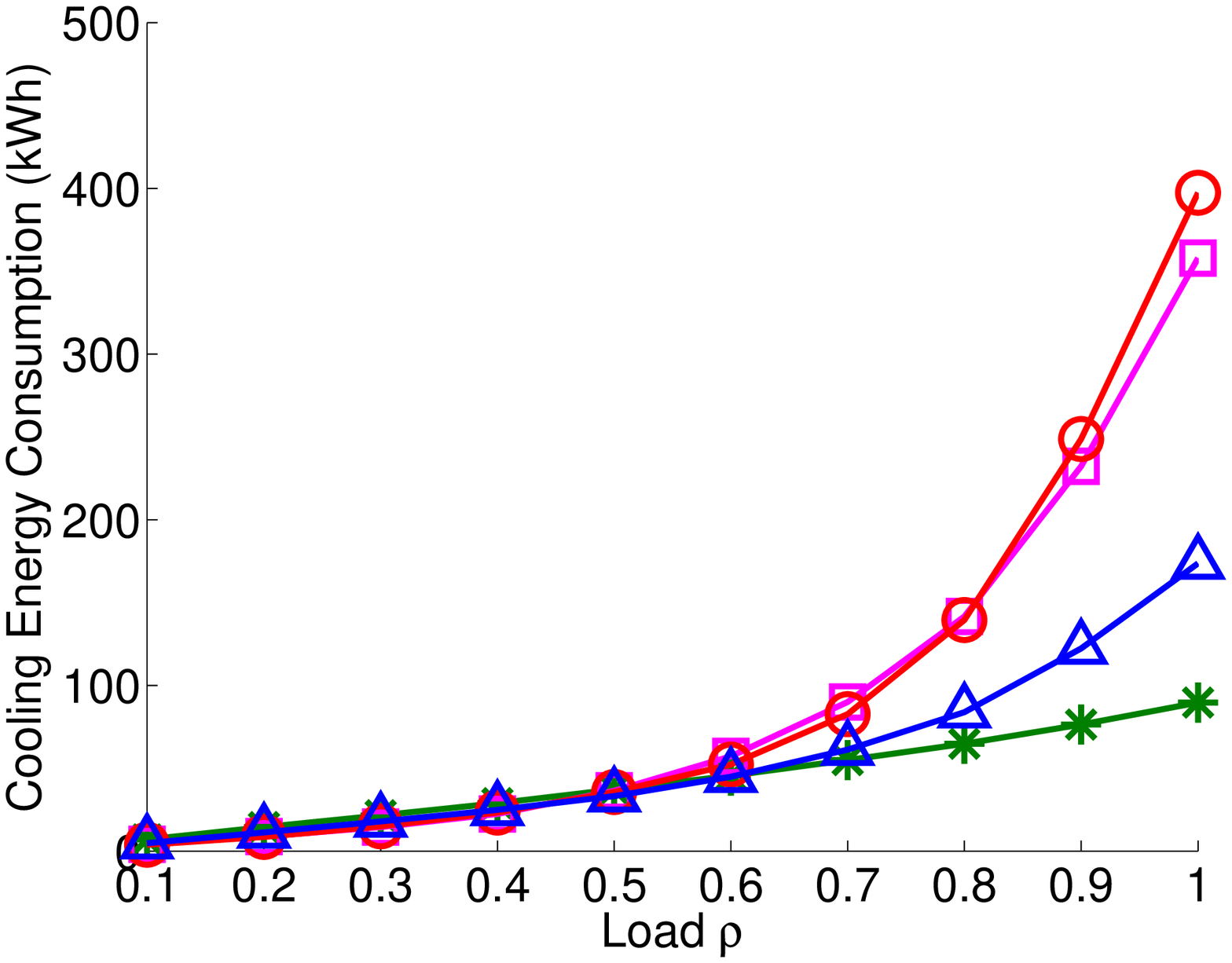}
    }
    \caption{Performance of \emph{Thermal-Aware} under different server placements and system loads. The legend applies to all subfigures.}
    \label{fig:sp_thermal}
\end{figure*}

\begin{figure*}[t]
\centering
    \subfigure[]
    {
        \label{fig:sp_minHR:meanRT}
        \includegraphics[width=2.1in]{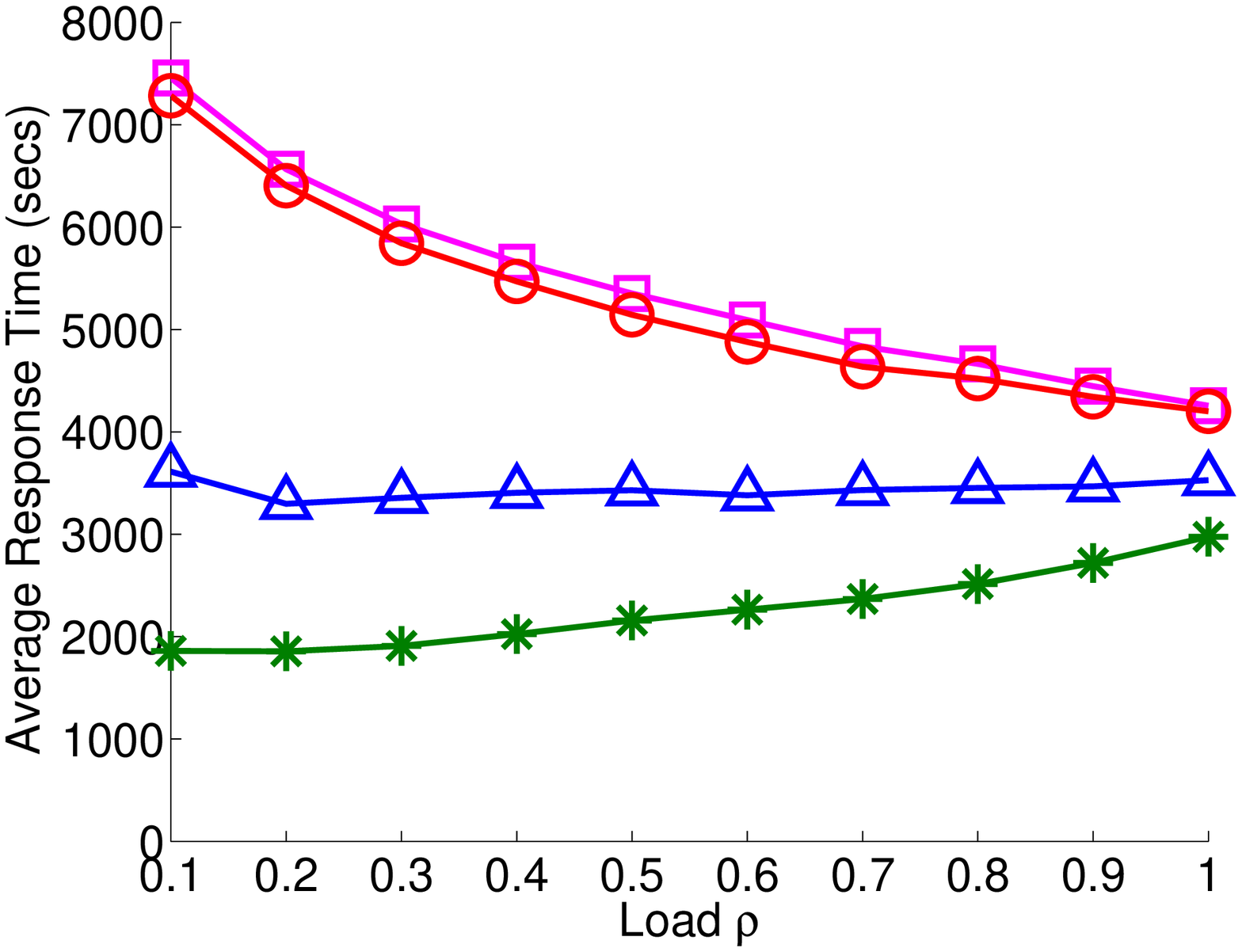}
    }
    \subfigure[]
    {
        \label{fig:sp_minHR:computingEnergy}
        \includegraphics[width=2.1in]{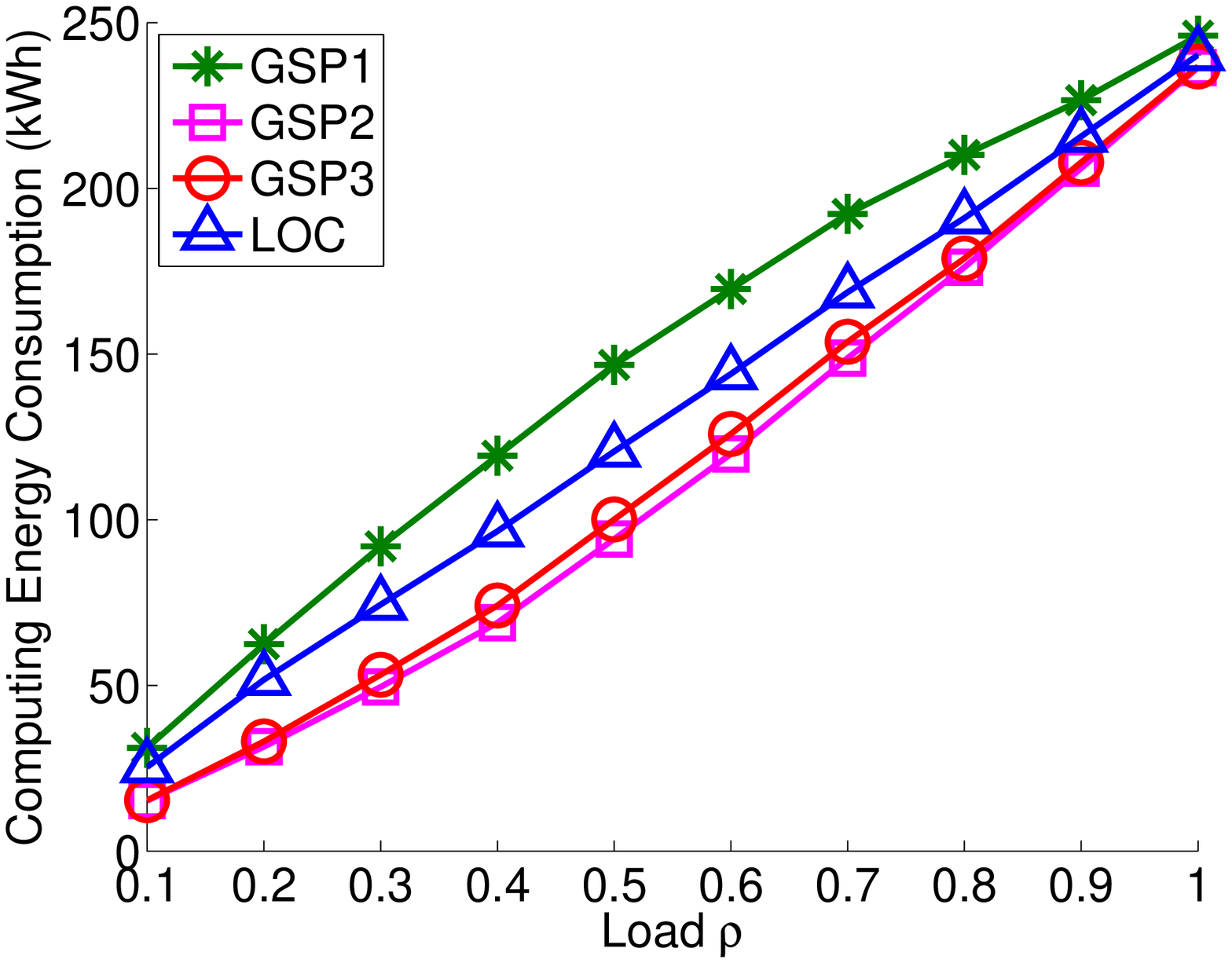}
    }
    \subfigure[]
    {
        \label{fig:sp_minHR:coolingEnergy}
        \includegraphics[width=2.1in]{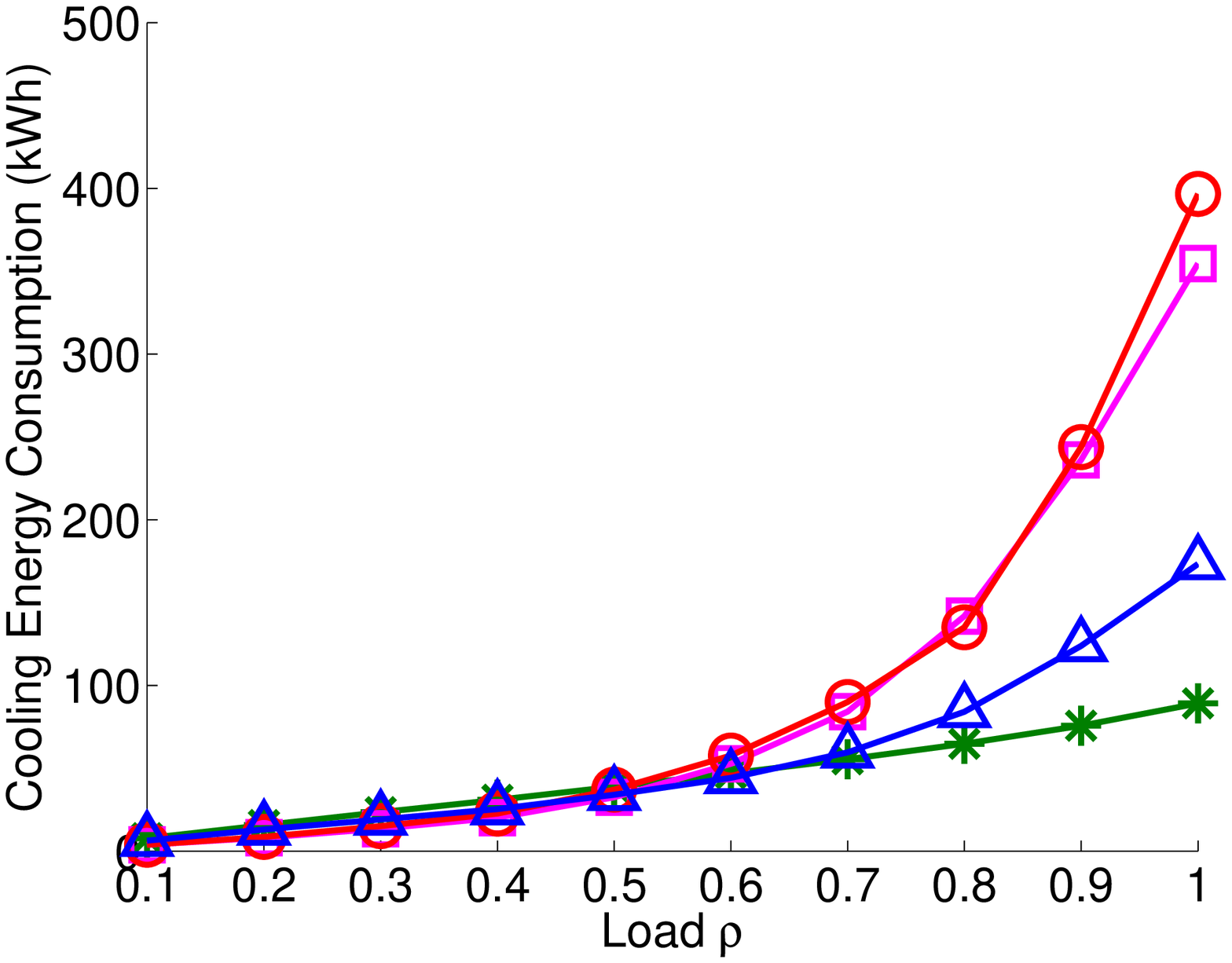}
    }
    \caption{Performance of \emph{MinHR} under different server placements and system loads. The legend applies to all subfigures.}
    \label{fig:sp_minHR}
\end{figure*}

\begin{figure*}[t]
\centering
    \subfigure[$H_{i, j}^{E, P}$]
    {
        \label{fig:tradeoff_mix:energy}
        \includegraphics[width=2.1in]{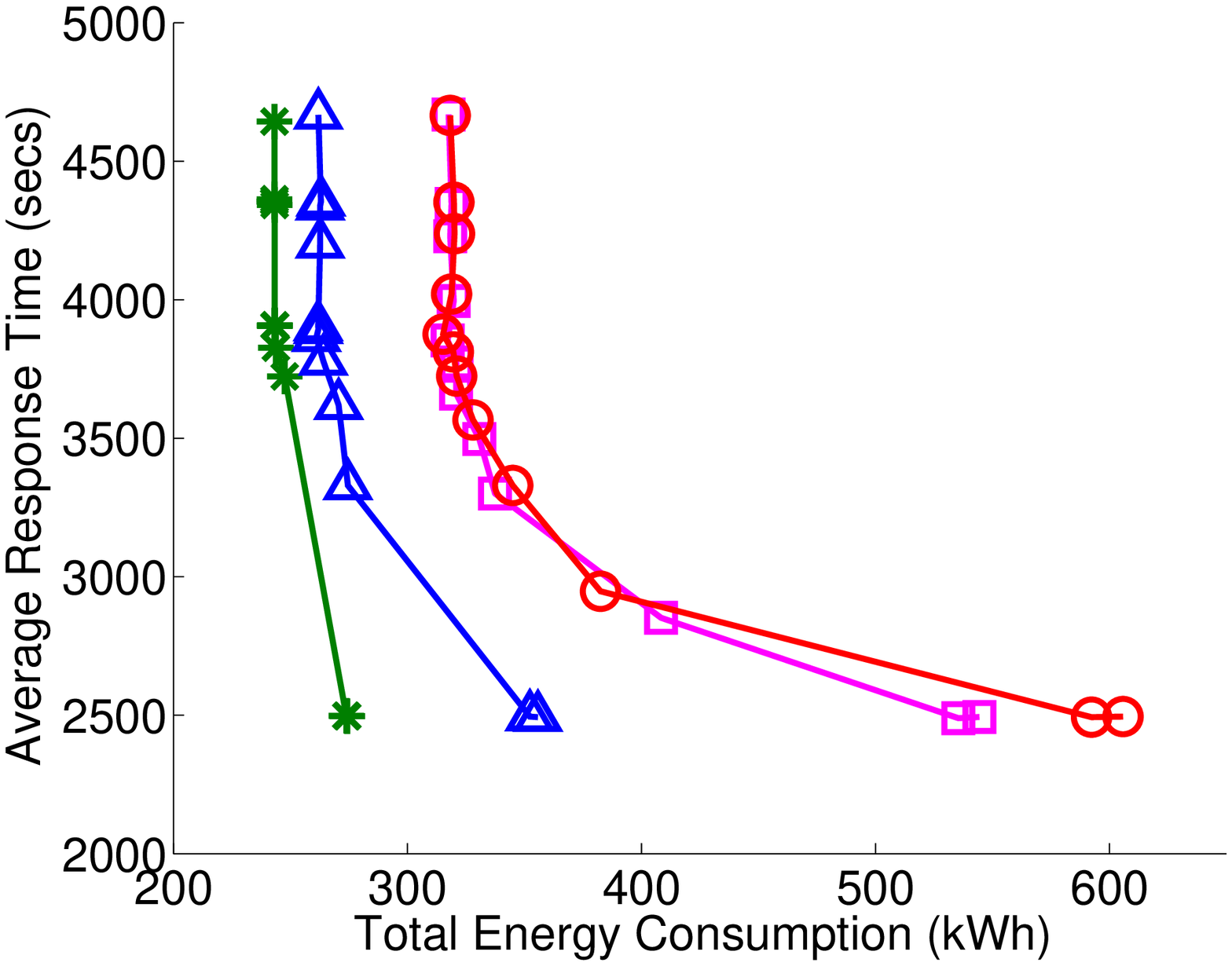}
    }
    \subfigure[$H_{i, j}^{HR, P}$]
    {
        \label{fig:tradeoff_mix:minHR}
        \includegraphics[width=2.1in]{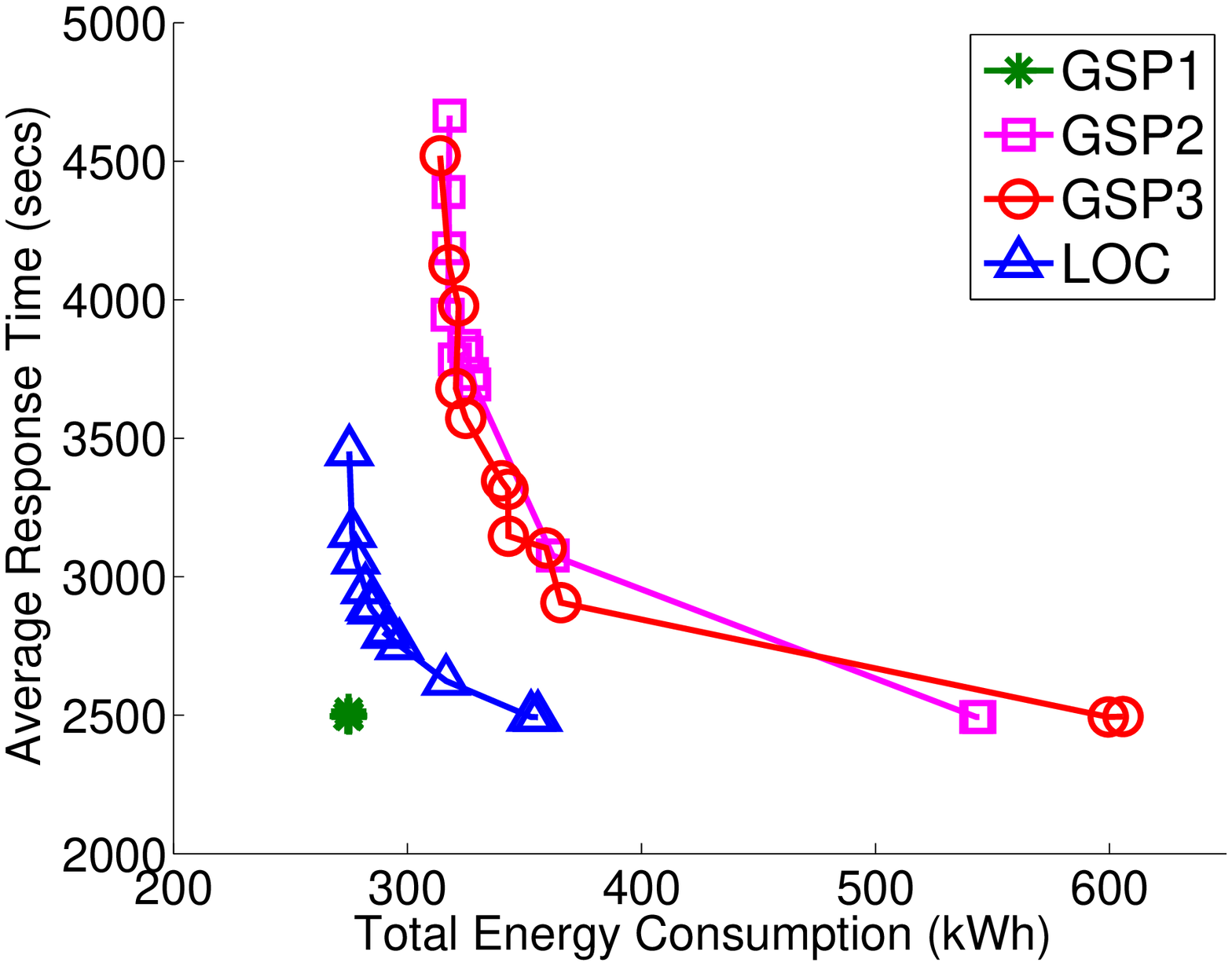}
    }
    \subfigure[$H_{i, j}^{T, P}$]
    {
        \label{fig:sp_tradeoff_mix:fast}
        \includegraphics[width=2.1in]{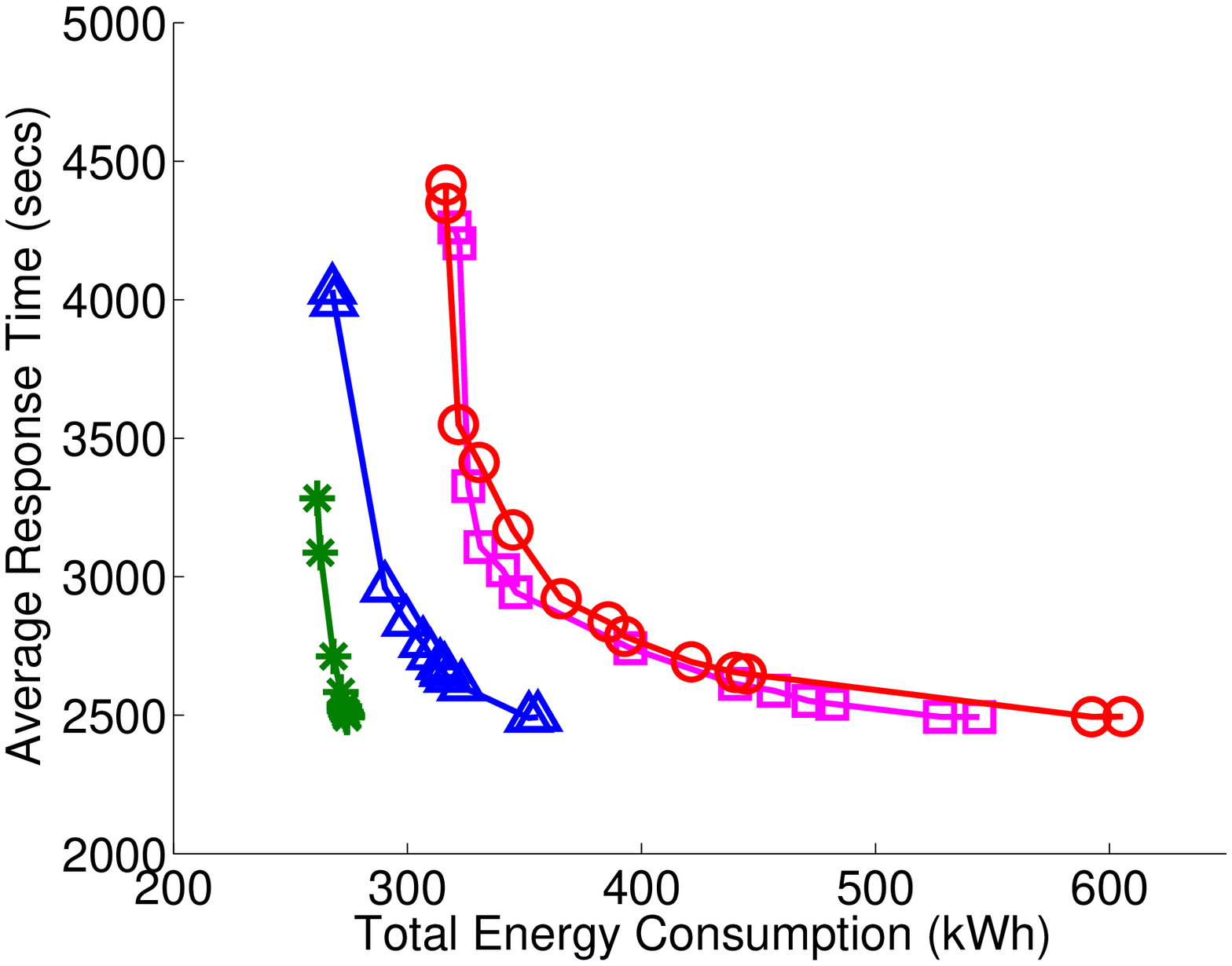}
    }
    \caption{Energy-performance tradeoff curves for $H_{i, j}^{E, P}$, $H_{i, j}^{HR, P}$ and $H_{i, j}^{T, P}$ under four different server placements at load $\rho = 0.8$. The legend applies to all subfigures.}
    \label{fig:tradeoff_mix}
\end{figure*}

\section{Related Work}\label{sec:related}

In this section, we review some related work in the literature on multi-objective scheduling and
thermal-aware scheduling for datacenters.

\paragraph{Multi-objective scheduling}

Scheduling with multiple conflicting objectives has attracted much attention in many optimization
problems. \secref{multi} described a few commonly used approaches. The following reviews some
applications of these approaches in various problem domains.

(1). \emph{Simple priority.} This is a simple priority-based approach to optimize multiple
objectives in sequence. Assayad et al.~\cite{Assayad04_MakespanReliabilityCompromiseFunction}
introduced a bi-criteria compromise function to set priorities between makespan and reliability for
scheduling real-time applications. To minimize carbon emission and to maximize profit, two-step
policies were proposed by Garg et al. \cite{Garg11_ProfitEnergyTradeoff} to map applications to
heterogeneous datacenters based on the relative priority of the two objectives. Du et
al.~\cite{Du13_Tradeoff} proposed heuristics to optimize the QoS for interactive services before
considering energy consumption on multicore processors with DVFS (Dynamic Voltage \& Frequency
Scaling) capability.

(2). \emph{Pareto frontier.} This approach is often used in the offline setting to generate a set
of nondominated solutions. Durillo et al. \cite{Durillo13_ParetoFront} applied this technique to
tradeoff makespan and energy consumption for heterogeneous servers. Torabi et al.
\cite{Torabi13_PSOWithFuzzyData} used particle swarm optimization to approximate the pareto
frontier for the unrelated machine scheduling problem with uncertainties in the inputs. Gao et
al.~\cite{Gao13_VirtualMachines} utilizes ant colony optimization to obtain the pareto frontier for
resource wastage and power consumption in virtual machine placement. Evolutionary algorithms were
employed in \cite{Yu07_GridEAParetoFront, Garg12_EAParetoFront} to obtain a set of alternative
solutions for scheduling scientific workloads in the Grid environment.

(3). \emph{Constraint optimization.} This approach optimizes one objective subject to constraint(s)
on the other(s). Rizvandi et al.~\cite{Rizvandi10_LinearCombineDVFS} applied it to minimize the
energy consumption subject to the makespan achieved in an initial schedule. A mixed integer
programming model was used by Petrucci et al.~\cite{Petrucci11_EnergySubjectToPerformance} to
reduce the power consumption of virtualized servers subject to QoS requirements. Fard et al.
\cite{Fard12_UserSpecifiedConstraints} developed a double strategy to minimize the Euclidean
distance between the generated solutions to a set of user-specified constraints in a four-objective
optimization problem. The authors in \cite{Grandinetti13_eConstraintMethod} applied
$\epsilon$-constraint method to cloud scheduling, which optimizes each objective in turn with upper
bounds specified for the others.

(4). \emph{Weighted combination.} This approach combines multiple objectives into a single one. Lee
and Zomaya \cite{Lee11_NormalizedTimePlusEnergyThenEnergySubjectToTime} used DVFS to tradeoff
makespan with energy consumption by considering a weighted sum of the two objectives. The same
technique was used by the authors of \cite{Andrew10_WeightedSum, Sun09_Tradeoff} in an online
manner to minimize a combined objective of job response time and energy. A similar approach was
taken by Sheikh and Ahmad \cite{Sheikh12_parentoSetThenWeightedSum}, who considered an additional
objective of peak temperature in a multicore system, and hence optimizing three objectives at the
same time. Instead of summation, some work (e.g., \cite{Cong12_EnergyDelayProduct,
Petrucci12_EnergyDelayProduct}) also used energy-delay product as a metric for scheduling
applications in heterogeneous multicore systems.

Compared to these approaches, our fuzzy-based priority approach provides a rather flexible solution
to handling two or more conflicting objectives. Although multi-objective scheduling with ``fuzzy"
or ``good enough" solutions \cite{Xu10_MultiFuzzySet, Zhang13_OOMethod} are known in the
pareto-based approach, our fuzzy method is novel when (soft) priorities exist between different
objectives. The principle can be potentially applied to other multi-objective optimization
problems.

\paragraph{Thermal-Aware Scheduling}

As cooling energy constitutes a signifcant fraction of the total energy consumption in today's
large-scale datacenter, thermal-aware scheduling for this environment has been the focus of many
research in recent years.

Wang et al.~\cite{Wang09_ThermalDataCenter, Wang12_ThermalDatacenter} considered thermal-aware
workload placement in datacenters to reduce the server temperatures characterized by an RC-model,
while minimizing the job response time. They proposed simple heuristics that allocate ``hot" jobs
to ``cool" computing nodes, as well as backfilling techniques for scheduling parallel applications.
In their study, the thermal map of the data center is assumed to be available through ambient and
on-board sensors.

Moore et al.~\cite{Moore05_Datacenter} first introduced the concept of heat recirculation effect
and proposed workload placement algorithms, including \emph{MinHR}, to reduce the recirculation of
heat and the cooling cost in a datacenter. A prediction tool called Weatherman
\cite{Moore06_Weatherman} was used to predict the data center thermal map using machine learning
techniques. The authors showed that the tool accurately predicts the heat distribution of the
datacenter without the need of static thermal configuration, and a scheduling algorithm based on
Weatherman achieves similar performance as \emph{MinHR}.

Tang et al.~\cite{Tang08_CyberPhysical} also studied the problem of minimizing the cooling cost in
datacenters with heat recirculation consideration. Based on an abstract heat flow model, they
characterized the thermal behavior of datacenters via a heat distribution matrix. The model was
validated by computational fluid dynamics (CFD) simulations in \cite{Tang06_SensorBased,
Sansottera11_CoolingPerformance}. They proposed offline scheduling solutions by using genetic
algorithms and quadratic programming, which were evaluated using the heat distribution matrix
captured for a small-scale datacenter. The same matrix is used in this paper for evaluating our
online scheduling heuristics.

Instead of minimizing only the cooling cost, Pakbaznia and Pedram
\cite{Pakbaznia09_TotalComputeAndCoolingPower} considered minimizing the total energy of a
datacenter from both computation and cooling. They showed that performing consolidation to turn off
idle servers together with job scheduling to account for the heat recirculation can significantly
reduce the total power usage. Banerjee et al.~\cite{BanerjeeMuVa11_CoolingAware} further
considered cooling-aware scheduling workload placement by exploring the dynamic cooling behavior of
the CRAC unit in a datacenter.

While the above results considered only the energy consumption of a datacenter, the following also
takes application performance into consideration. Mukherjee et al.~\cite{Mukherjee09_Datacenter}
considered a similar problem as in \cite{Pakbaznia09_TotalComputeAndCoolingPower} but further took
the temporal dimension of the job placements into account. They formulated the problem as a
non-linear program and proposed both offline and online heuristics to minimize the total energy
subject to the deadline constraint for the jobs. Sansottera and Cremonesi
\cite{Sansottera11_CoolingPerformance} considered a datacenter environment hosting web services,
and presented heuristics to minimize the total energy subject to service response time constraints.
Kaplan et al. \cite{Kaplan13_Communication} studied the dual optimization of cooling and
communication costs for HPC applications in a datacenter. They proposed a heuristic algorithm that
achieves a good tradeoff between the two objectives, and subject to reliabity constraint specified
by the processor junction temperature.

In contrast to the previous work, which focused on either offline scheduling or homogeneous
datacenters, we studied the problem of online scheduling for heterogeneous datacenters with both
energy and performance considerations, as well as their tradeoffs. Furthermore, we considered
static server placement to balance the thermal distribution in the presence of nonuniform heat
distribution matrix. In our previous work \cite{Sun14_fuzzyFactor}, we have applied this concept to
the arrangement of computing nodes in a smaller scale problem (at the server level). To the best of
our knowledge, no prior work has considered this problem for heterogeneous datacenters.

\section{Conclusion and Future Work}\label{sec:conclusion}

In this paper, we have considered the energy-efficient and thermal-aware placements for both
servers and workloads in heterogeneous datacenters. For the static server placement problem, we
have shown that it is NP-hard and presented a greedy heuristic. To schedule the workloads, we have
presented a greedy scheduling framework, which can be applied in an online manner with any
well-defined cost function. Moreover, a novel fuzzy-based priority approach was proposed to
simultaneously optimize two or more conflicting objectives. Simulations were conducted for a
heterogeneous datacenter with heat recirculation effect. The results demonstrated the effectiveness
of the proposed approaches for exploring the energy-performance tradeoff with cooling
consideration. Our static server placement heuristic was also shown to provide better thermal
balance, which directly leads to reductions in cooling cost.

For future work, other resource management techniques, such as DVFS or server consolidation, can be
applied to achieve better energy and thermal efficiency. In this context, the tradeoff between the
computing energy and cooling energy can be explored, possibly with the fuzzy-based priority
approach. For the static server placement problem, it will be useful to design better heuristic
solutions or good approximation algorithms, and to consider large datacenters with more rack slots
than servers, which will provide additional space for optimization. Finally, we considered server
placement and job scheduling separately in this paper; it may be helpful to consider the two
aspects jointly to achieve further energy savings.

\section*{Acknowledgment}

This research was funded by the European Commission under contract 288701 through the project
CoolEmAll.

\bibliographystyle{IEEEtran}

\section*{Appendix. NP-Hardness Proof of the Static Server Placement Problem}

\begin{claim}
The static server placement problem is NP-hard.
\end{claim}

\begin{proof}
We reduce the 3-partition problem to the static server placement problem. In 3-partition, a finite
set $\A = \{a_1, a_2, \cdots, a_n\}\subset \mathbb{Z^+}$ of $n = 3h$ positive integers is given,
and the sum of the integers is $\sum_{j=1..n} a_j = h\cdot B$. The question is whether $\A$ can be
partitioned into $h$ disjoint subsets such that the sum of the numbers in each subset is equal to
$B$. The problem is known to be NP-hard even if every integer in $\A$ is strictly between $B/4$ and
$B/2$, so each subset must contain exactly three numbers \cite{Garey79_NPHard}.

Given an instance $\A$ of the 3-partition problem, where each integer $a_j\in \A$ satisfies $B/4 <
a_j < B/2$, we construct an instance of the static server placement problem as follows. Let $m = n
= 3h$, and assign
\begin{equation*}
U_j^{ref} = a_j \ , \forall j = 1,\cdots,n \ .
\end{equation*}
The heat-distribution matrix $\mathbf{D}$ is specified by setting
\begin{equation*}
d_{3l, 3l-2} = d_{3l, 3l-1} = d_{3l, 3l} = 1 \ , \forall l = 1,\cdots,h \ ,
\end{equation*}
and setting all the other elements to zero.

Suppose $\sigma^*$ is an optimal mapping for the server placement instance constructed above. The
temperature increase at the inlet of slot $S_k$, where server $M_{\sigma^*(k)}$ is placed, is given
by
\[
    T^{incr}_{k}=
\begin{cases}
    a_{\sigma^*(k-2)} + a_{\sigma^*(k-1)} + a_{\sigma^*(k)},& \text{if } k\text{ mod }3 = 0\\
    0,& \text{otherwise }
\end{cases}.
\]
The maximum temperature increase at any inlet is therefore
\begin{equation*}
T^{incr}_{max} = \max_{k = 3, 6,\dots, 3h} \paren{a_{\sigma^*(k-2)} + a_{\sigma^*(k-1)} +
a_{\sigma^*(k)}} \ .
\end{equation*}
This leads to the conclusion that the server placement instance has a maximum temperature increase
of $B$ if and only if $\A$ can be partitioned into $h$ disjoint subsets, where the sum of the
numbers in each subset is also $B$.
\end{proof}

\end{document}